%% file: main.tex
\pdfoutput=1
\documentclass[runningheads]{llncs}

\usepackage[T1]{fontenc}
\usepackage{amsfonts,amsmath,amssymb}
\usepackage{thm-restate}   

\usepackage{graphicx,xcolor}
\usepackage{capt-of}      
\usepackage{algorithm}
\usepackage{algpseudocode}
\usepackage{booktabs}

\usepackage{tikz}
\usetikzlibrary{quantikz}

\usepackage[hidelinks]{hyperref}

\newcommand{\F}{\mathbb{F}}
\newcommand{\xor}{\oplus}
\newcommand{\CD}{\mathcal{MACD}}
\newcommand{\stdnorm}{\tfrac{1}{\sqrt{N}}}
\newcommand{\N}{\mathcal{N}}

\newcommand{\E}{\mathbb{E}}
\newcommand{\Var}{\operatorname{Var}}
\newcommand{\adv}{\operatorname{adv}}
\newcommand{\bad}{\mathtt{LOWSEEN}}
\newcommand{\cor}{\operatorname{Cor}}
\renewcommand{\Cap}{\operatorname{Cap}}

\begin{document}

\title{Masked Differential-linear Distinguishers and Quantum Approaches}
\titlerunning{Masked Differential-linear Distinguishers}

\author{Shobhit Pandey \and
Sarbani Sen \and
Debajyoti Bera \and
Ravi Anand}
\authorrunning{S. Pandey et al.}

\institute{Indraprastha Institute of Information Technology, Delhi, India}

\maketitle

\begin{abstract}
We introduce \emph{masked auto-correlation}, a new primitive for the
cryptanalysis of symmetric-key primitives, together with a quantum attack
pipeline built on it. For a permutation $f$, output masks $\alpha,\beta$, and an
input difference $w$, masked auto-correlation (MAC) measures the correlation
between the masked outputs $\alpha\cdot f(x)$ and $\beta\cdot f(x\oplus w)$. The
associated masked differential-linear (MDL) approximations strictly generalize
several classical techniques; ordinary linear cryptanalysis, differential-linear
cryptanalysis, and the differential-linear connectivity table all arise as
special cases.
Our central object of study is the problem of \emph{finding} mask pairs with
large masked cross-correlation---those that yield powerful
distinguishers---which we call MAC Fishing. We give a constant-query quantum
algorithm that samples such pairs according to their squared correlation, and we
prove an exponential classical lower bound of
$\Omega(N/\log N)$ queries, by adapting the hardness of Fourier Fishing. To our
knowledge this is the first result pairing a quantum upper bound with a classical
lower bound for the core task of identifying high-correlation approximations,
making quantum algorithms an absolute necessity.
Building on this, we analyse the distribution of masked auto-correlation for
random permutations, and then construct capacity-based
distinguishers and key-recovery attacks, both classically and with a quadratic
quantum speed-up using amplitude estimation. We validate our claims with
experiments on reduced-round mini-AES.

\end{abstract}

\input{final/introduction}
\input{final/background}
\input{final/masked_auto_cor}

\input{final/mac_fishing}
\input{final/mdl_approx}

\input{final/experiments}

\input{final/attacks}

\bibliographystyle{splncs04}
\bibliography{references}

\appendix
\input{final/appendix_lower-bound-proof}

\input{final/appendix}

\input{final/appendix-simple-proofs}

\end{document}

%% file: final/introduction.tex
\section{Introduction}

The security of symmetric-key primitives---block ciphers, stream ciphers, and
their associated modes---has been studied extensively in the classical paradigm.
The two foundational classical attack paradigms, linear
cryptanalysis~\cite{matsui1993linear} and differential
cryptanalysis~\cite{biham1991differential}, were introduced in the early
1990s and remain the primary benchmarks against which new designs are evaluated.
Subsequent work has produced a rich body of classical techniques that interpolate
and extend these two paradigms: differential-linear
cryptanalysis~\cite{langford1994differential}, multiple and multidimensional
linear approximations~\cite{HCN09,HCN19}, and algebraic
variants~\cite{courtois-meier-algebraic}, among others. Despite this
diversity, virtually all of these techniques share a common classical
algorithmic trait: the attacker processes keystream or ciphertext samples
sequentially, accumulates statistical evidence, and constructs a distinguishing
attack or recovers key. In the classical setting, the bottleneck is
always the \emph{number of plaintext--ciphertext pairs} required to drive the
distinguishing attack to significance, and the design community has
responded by designing ciphers whose linear and differential probabilities are as
low as possible.

The advent of quantum computing forces a re-examination of this
picture. Quantum algorithms can offer quadratic to superpolynomial speedups for
certain structured search and sampling problems, so it is natural to ask whether
the statistical properties underlying symmetric-key cryptanalysis admit similar
quantum speed-ups. Early results showed that
Simon's algorithm can break many constructions with a secret period
structure~\cite{simon-algorithm,kuwakado-morii-simon}, and subsequent work has
turned classical statistical attacks---linear, differential, and
meet-in-the-middle---into the quantum
setting~\cite{boneh-zhandry-q1q2,kaplan-q1q2,hosoyamada-sasaki-quantum-linear}.
Most recently, Hosoyamada studied quantum multi-linear
approximations~\cite{hosoyamada2023multidim}, showing that the quantum
query complexity of multidimensional linear distinguishers can be significantly
lower than their classical counterparts. Despite this progress, most existing
quantum attacks inherit the structure of their classical counterparts,
porting a known statistical distinguisher into a quantum query model rather than
identifying \emph{new} properties that are inherently quantum in nature.
In this paper we introduce a new cryptanalytic primitive, \emph{masked
auto-correlation}, whose most powerful instantiation is provably beyond the reach
of efficient classical algorithms.

\paragraph{Masked auto-correlation.}
At the base of linear cryptanalysis is the Walsh--Hadamard transform: for
$f:\F_2^n\to\F_2$, the coefficient $\widehat f(\lambda)$ measures the correlation
between $f$ and the linear function $x\mapsto\lambda\cdot x$, and the design goal
is to keep all coefficients small. \emph{Cross-correlation} generalizes this:
instead of correlating $f$ against a fixed linear function, one correlates it
against a shifted copy of another function, sweeping over all shifts, and
\emph{auto-correlation} is the special case of a function against a shifted copy
of itself. We study this in the vectorial setting natural for block ciphers,
where the round function is a permutation $f:\F_2^n\to\F_2^n$, and ask where the
auto-correlations among different output bits are statistically significant.
Rather than fix a single output bit, we take linear combinations: fixing output
masks $\alpha,\beta\in\F_2^n$, the \emph{masked auto-correlation} measures the
cross-correlation between $\alpha\cdot f(x)$ and $\beta\cdot f(x\oplus w)$ as the
difference $w$ ranges over $\F_2^n$,
\[
  C^{w}_{f}(\alpha,\beta) \;\propto\; \sum_x (-1)^{\alpha\cdot f(x)\,\oplus\,\beta\cdot f(x\oplus w)}
\]
(formalized in Section~\ref{sec:mac}). Since $f$ is a permutation,
$C^w_f(\alpha,\beta)$ is proportional to the expectation of a $\pm1$ random
variable, so it can be estimated accurately for any \emph{given} $\alpha,\beta,w$.
Two obstacles stand in the way of exploiting it, however. First, \emph{finding}
the mask pairs $(\alpha,\beta)$ whose cross-correlation is large---those that
yield powerful distinguishers---is a search over an exponential space, infeasible
by naive enumeration. Second, even granted such masks, \emph{sampling} according
to the cross-correlation distribution has no obvious efficient classical solution.
The central thesis of this paper is that both obstacles fall to a quantum
algorithm, and that the resulting pipeline is provably beyond efficient classical
reach.

\subsection{Technical Overview}\label{sec:overview}

Our development rests on three objects, defined formally in
Sections~\ref{sec:mac} and~\ref{sec:mdl}: masked cross-correlation
$C^w_{f,g}(\alpha,\beta)$ of two permutations and its auto-correlation
specialization $C^w_f(\alpha,\beta)$ denoted {\em MAC}; the induced \emph{masked auto-correlation
distribution} $\CD_w$, which places mass $\propto C^w_f(\alpha,\beta)^2$ on
each mask pair, and the \emph{masked differential-linear (MDL)
approximation} $\alpha\cdot f(x)\oplus\beta\cdot g(x\oplus w)$, a two-function,
difference-aware generalization of the classical linear approximation. We
summarize our contributions built on them.

\paragraph{A quantum algorithm for finding good masks.}
Most prior work on linear and differential approximations \emph{assumes} a
highly-correlated approximation is already in hand~\cite{kaplan2016diff_lin}.
We instead give a quantum algorithm that \emph{finds} one: using $3n$ qubits and
two oracle calls to $f$, it prepares a superposition that samples from $\CD_w$ for a fixed key, so
a measurement returns a mask pair $(\alpha,\beta)$ with probability proportional
to $C_w(\alpha,\beta)^2$ --- heavy pairs, exactly the exploitable ones, are the
likely outputs. The algorithm is easily generalizable to sample across all keys with minimal overhead.

\paragraph{Classical hardness of mask finding.}
We then ask whether this search can be done classically. We formalize it as the
\emph{MAC Fishing} problem---output any $(\alpha,\beta)$ with large $|C_w(\alpha,\beta)|$, under the promise that such a pair exists---and prove it requires
$\Omega(N/\log N)$ classical queries. The argument adapts the Aaronson--Chen
hardness proof for \emph{Fourier Fishing}~\cite{aaronson-chen}, the canonical
example of a search that is easy quantumly yet provably hard classically and a
member of the \emph{forrelation} family of maximal query separations; masked
cross-correlation is a natural cross-correlation analogue.
To our knowledge this is the first provably exponential speedup shown 
for the central task of identifying high-correlation approximations.

As part of the proof, we derived sharp concentration bound on a possible statistic based on masked auto-correlation, namely, the total mass of high auto-correlations (those that cross a threshold). It allows us to show later that the statistic indeed differs widely in Mini-AES as compared to that in a random permutation.

\paragraph{A unifying framework and its applications.}
The MDL approximation is a strict generalization of several classical
distinguishers rather than a technique alongside them. Taking $g$ the identity
and $w=0^n$ recovers ordinary linear cryptanalysis; enforcing identical masks
$\alpha=\beta$ recovers differential-linear cryptanalysis and, in the single-bit
case, the DLCT entry; and organizing mask pairs into a vector space recovers the
multidimensional variant of each (Section~\ref{sec:mdl-relations}). What unifies
them is not only the form of the approximation but its \emph{statistics}: in
every case the figure of merit is the capacity of a space of masked
auto-correlations, which we show obeys the same $\chi^2$ law for a random
permutation (Conjecture~\ref{conj:capacity}). A single distributional analysis
therefore governs the whole family, and a single distinguisher exploits it.

On top of this we build an end-to-end pipeline. Classically, a capacity-based
distinguisher thresholds the aggregated masked auto-correlation of a basis of
mask pairs, and a last-round key recovery extends it: a wrong subkey guess
randomizes the partially decrypted cipher, so the correct key is the one whose
correlation profile matches a reference template. Crucially this match is
invisible to the unsigned energy---the subkey enters the masked bits as a
sign---so recovery is resolved by \emph{signed} template matching
(Section~\ref{sec:applications}). Quantumly, both attacks are driven by \emph{MAC
Fishing}: the masks are supplied by the constant-query sampler of
Section~\ref{sec:macfishing} and each correlation is read off by amplitude
estimation, replacing the classical $O(1/\tau)$ verification with
$O(1/\sqrt\tau)$ for a quadratic query advantage---on top of a mask-finding step
with no efficient classical analogue at all.

We validate the two halves on the two ciphers each suits best. The distributional
claims are verified on a $16$-bit 3-round Mini-AES, whose block is small enough to
measure every correlation exactly over the full codebook and to estimate these
distributions across many keys against a known baseline; further, we used a ``small'' shift of $w=0x0003$ that essentially flips the two rightmost bits, yet, we were surprised to see the amount of auto-correlation produced by it. The other half, the attacks, are exercised on a toy Simon cipher, whose Feistel last round inserts the subkey as a
constant phase on the masked output bits and so realizes, in clean form, the
sign-degeneracy that the signed-template key recovery is designed to resolve
(Appendix~\ref{sec:exp-validation}).

\subsection{Our Contributions}\label{sec:contributions}
\begin{itemize}
  \item We extend cross-correlation to vectorial
    Boolean functions and define (multidimensional) MDL approximations
    (Section~\ref{sec:mdl}).  Linear cryptanalysis, differential-linear
    cryptanalysis, their multidimensional variants, and the DLCT are each strict
    special cases, obtained by constraining the output masks or the input
    difference.

  \item We characterize masked auto-correlation
    for a random permutation (Section~\ref{sec:mac}), and conjecture that the capacity of a $t$-dimensional MDL
    approximation with non-trivial masks (i.e. masks \textit{not} of the form $(\alpha, 0)$ or $(0, \beta)$) follows a central $\chi^2$ with $2^t-1$
    degrees of freedom (Conjecture~\ref{conj:capacity}).

  \item We give a
    constant-query quantum algorithm that samples mask pairs in proportion to
    their squared cross-correlation, directly producing high-correlation MDL
    approximations, and a matching classical lower bound: {\em MAC Fishing} requires
    $\Omega(N/\log N)$ classical queries (Theorem~\ref{thm:MAC_lb}), adapting
    the Aaronson--Chen hardness of Fourier Fishing. To our knowledge this is the
    first pairing of a quantum upper bound with a classical lower bound for
    identifying high-correlation approximations.

  \item We build a capacity-based
    distinguisher and a last-round key recovery, the latter using signed template
    matching to resolve a sign-degeneracy invisible to capacity ranking, and give
    {\em MAC-Fishing} driven quantum versions in which the masks come from the quantum
    sampler and each correlation is verified by amplitude estimation, for a
    quadratic query advantage (Section~\ref{sec:applications}).

  \item We verify the framework on the two
    ciphers each part suits. On $16$-bit 3-round Mini-AES we confirm the distributional
    theory: heavy mask pairs exist for the vast majority of keys, the good ones
    are stable across keys, and the random-permutation capacity matches the
    predicted $\chi^2$ (Sections~\ref{sec:macfishing} and~\ref{sec:mdl}, Figures~\ref{fig:v_tau},\ref{fig:mini_AES_3round_example_masks} and \ref{fig:capacity_dist}). On a toy
    Simon cipher we run the pipeline end to end---the capacity distinguisher
    separates the cipher from a random permutation at the predicted data
    complexity, and the signed template recovers the last-round subkey where the
    unsigned energy cannot (Section~\ref{sec:applications},
    Appendix~\ref{sec:exp-validation}).
\end{itemize}
\begin{remark}[Why two ciphers]
    Mini-AES and toy Simon are doing different jobs, and each fits its job structurally. The \textit{distributional} claims — that heavy mask pairs exist, are stable across keys, and that random-permutation capacity is $\chi^2$  — are statements about a population of keys and masks, so they need many keys each measured exactly. Mini-AES's 16-bit block makes that feasible: every correlation is computed over the full codebook and the random baseline sits at a known $\sigma=2^{-8}$. The \textit{attacks}, by contrast, hinge on the last-round structure. Simon's Feistel round inserts the subkey into the masked output bits as a constant XOR, so $\cor_i(k)=(-1)^{\langle k,m_i\rangle}\cor_i^\star$  with key-mask $m_i=\alpha_i\oplus\beta_i$ — the magnitudes are key-blind and all the key information lives in the signs, which is exactly the sign-degeneracy the signed-template recovery is built to resolve. Mini-AES's nonlinear last round gives no such clean phase, so it's the wrong structure for our key-recovery demonstration. Hence: Mini-AES verifies the statistics, Simon exercises the pipeline.
\end{remark}
\subsection{Organization}\label{sec:organization}
Section~\ref{sec:background} reviews correlation, capacity, linear and
differential-linear cryptanalysis, and the Fourier Fishing hardness result.
Section~\ref{sec:mac} develops the properties of masked auto-correlation for an
arbitrary and for a random permutation. Section~\ref{sec:macfishing} introduces
the MAC Fishing problem, the constant-query quantum sampler, and the
classical-hardness theorem (proved in Appendix~\ref{sec:proof-MACfishing-lb}), and
shows experimentally that heavy mask pairs are abundant and independent of the key in
Mini-AES. Section~\ref{sec:mdl} defines MDL approximations, shows how they unify
multidimensional linear, differential-linear, and DLCT cryptanalysis, and
verifies the $\chi^2$ capacity law on random permutations.
Section~\ref{sec:applications} turns these into attacks---a capacity-based
distinguisher and last-round key recovery, demonstrated on Mini-AES, and their
MAC-Fishing-driven quantum counterparts. Appendix~\ref{sec:exp-validation}
reports the end-to-end validation of the full pipeline on a toy Simon cipher; the
remaining appendices collect the proofs and the classical verification routines. Several of our experiments used Wilson's score to produce results with high confidence; a background on using this method is included in Appendix~\ref{app:Wilson}.

%% file: final/background.tex
\section{Background and Related Works}
\label{sec:background}
\subsection{Notation}
Throughout, $n$ is the block size and $N = 2^n$. We work over $\F_2^n$, with
$\cdot$ denoting the inner product over $\F_2$ and $\xor$ bitwise XOR. Lower-case
$\alpha,\beta,\lambda \in \F_2^n$ denote output masks and $w \in \F_2^n$ an input
difference (or shift); $0^n$ is the all-zeroes string. A permutation
$f ~\text{or}~B : \F_2^n \to \F_2^n$ models a (round-reduced) block cipher, and $E_K$ the
full cipher under key $K$, with $R_k^{-1}$ the inverse last-round map under a
subkey guess $k$. We write $C^w_f(\alpha,\beta)$ for the masked auto-correlation
of $f$, abbreviated $C_w(\alpha,\beta)$ when $f$ is clear from context, and
$\CD_w$ for the induced distribution over mask pairs. $\Cap(\cdot)$ denotes
capacity, $\cor(\cdot)$ correlation, and $\N(\mu,\sigma^2)$ the normal
distribution. We use standard asymptotic notation ($O,\Omega,\omega,o$). For a
value distribution on a codomain of size $T$ the capacity carries a multiplier
$T$; in particular a $\F_2^t$-valued function has multiplier $2^t$, which is the
form used in later sections.

\subsection{Correlation, capacity, and approximations}
The correlation of a Boolean function $f : \F_2^n \to \F_2$ is defined as its deviation from the uniform function; mathematically,
$$\cor(f) = \Pr_x[f(x)=0] - \Pr_x[f(x)=1] = \tfrac{1}{2^n}\sum_x (-1)^{f(x)}.$$


Next, let $f:\F_2^n \to \{1, \ldots t\}$. Value distribution ({\it aka.} probability distribution) of $f$ is defined as 
$p_f(y) = \Pr_{x \in_R \F_2^n} [f(x)=y]$ for $y \in \{1, \ldots, t\}$.

Capacity of a value distribution, say $p_f(y)$, was defined by Biryukov et al.~\cite{biryukov2004multiple}.
$$\Cap(p_f) = t \sum_{y} \left(p_f(y) - \tfrac{1}{t} \right)^2$$

Capacity is also known as ``squared Euclidean imbalance''. 
Biryukov et al. \cite{biryukov2004multiple} showed that the data complexity to distinguish between $f$ and a uniform random function is proportional to the inverse of its capacity, $O(1/\Cap(p_f))$.

Finally, consider the case of $f:\F_2^n \to \F_2^n$ being a permutation. Here, of course, $p_f(y) = 1/2^n$ since $f$ is a permutation.

A common approach undertaken for cryptanalysis of ciphers is to study one or more approximations, usually linear in the input and output of the cipher. Suppose, $g^B(x):\F_2^n \to \{1,\ldots, t\}$ is an approximation of a block cipher $B$ that has a bias $\epsilon$, i.e., $|\Pr_x[g^B(x)=0] - 1/2| = \epsilon$. Then, the capacity could act as a distinguisher between $B$ and a random function; e.g., it would be $2\epsilon^2$ for $B$. That implies, that the data complexity of this distinguisher would be $O(1/\epsilon^2)$.

Biryukov et al. demonstrated that defining multiple approximations and then using a suitable notion of capacity for multiple approximations may lead to a smaller data complexity~\cite{biryukov2004multiple,hermelin2008multidim}.

In this work we want to design an efficient distinguisher that is defined in terms of multiple approximations involving masked auto-correlations. However, before that, let's discuss the prominent earlier works along this direction.

\subsection{Linear cryptanalysis}

Linear cryptanalysis, introduced by Matsui~\cite{matsui1993linear}, distinguishes a
block cipher~$B$ from a random permutation by exhibiting a linear approximation
$\alpha\cdot x \oplus \beta\cdot B(x)$ whose correlation
\[
  \mathrm{Cor}(B;\alpha,\beta) \;:=\; \Pr_x[\alpha\cdot x = \beta\cdot B(x)] -
  \Pr_x[\alpha\cdot x \neq \beta\cdot B(x)]
\]
deviates noticeably from~$0$. For an $n$\nobreakdash-bit random permutation
$\mathrm{Cor}$ is approximately $\mathcal{N}(0,2^{-n})$ distributed for any non-zero
mask pair~\cite{DR07}, so an empirical correlation outside
$[-2^{-n/2},2^{-n/2}]$ already suffices to distinguish, with data complexity
$O(\mathrm{Cor}^{-2})$. Early attempts to combine several approximations relied on
(generally invalid) statistical-independence assumptions~%
\cite{kaliski1994multiple,biryukov2004multiple,murphy2006independence}.

\subsection{Multidimensional Linear Cryptanalysis}
\label{subsec:multidim-lin-crypt}

A common approach to designing a distinguisher from several approximations at
once is Pearson's chi-square goodness-of-fit test, so we first recall it. The
chi-square distribution arises as the distribution of a sum of squares of
independent standard normal variables. Formally, $\chi^2_\ell(\delta)$ denotes the
(non-central) chi-square distribution with $\ell$ degrees of freedom and
non-centrality parameter $\delta$, with mean $\ell + \delta$ and variance
$2(\ell + 2\delta)$. Its central role in cryptanalysis stems from Pearson's
goodness-of-fit test: given $k$ multinomial cell counts $x_i$ from $m$ trials
with expected probabilities $p_i$, the statistic
\begin{equation}
    T = \sum_{i=1}^{k} \frac{(x_i - m p_i)^2}{m p_i}
\end{equation}
is asymptotically $\chi^2_{k-1}(\delta)$-distributed with $\delta = \sum_{i=1}^{k}(m p_i)^2$~\cite{RT89}.

Building on this, Hermelin, Cho and Nyberg~\cite{hermelin2008multidim,hermelin2009multiext,HCN19}
proposed \emph{multidimensional linear cryptanalysis} to address the statistical-independence assumption inherent in the early works that combined multiple linear approxiations. Their idea was to select a set of input--output mask pairs $S=\{(\alpha_i,\beta_i)~:~i=1 \ldots l\}$ that form some basis of a vector
space $V\subset\mathbb{F}_2^m\times\mathbb{F}_2^n$. The induced distribution
\[
  p_S^{f}(z) \;:=\; \Pr_x\!\left[\bigl(\alpha_1\cdot x\oplus\beta_1\cdot f(x),\,\ldots,\,
                          \alpha_\ell\cdot x\oplus\beta_\ell\cdot f(x)\bigr) = z\right]
\]
on $\mathbb{F}_2^{\ell}$ deviates from the uniform distribution precisely when the capacity of $p^f_S$ is large.

Furthermore, since 
\[
  \mathrm{Cap}(p_S^{f}) \;=\; \sum_{(\alpha,\beta)\in V\setminus\{0\}}
                              \mathrm{Cor}(f;\alpha,\beta)^2
\]
and $\cor(f;\alpha,\beta)$ is distributed normally for a uniform random permutation, Pearson's $\chi^{2}$ goodness-of-fit test has been proposed as a distinguisher with data complexity
$O\!\bigl(\sqrt{2^{\ell}}/\mathrm{Cap}(p_S^{f})\bigr)$. Further improvements have been proposed, e.g., LLR-based
variants~\cite{baigneres2004how,cho2010linear} can reach $O(1/\mathrm{Cap})$ at the
cost of accurate knowledge of the (key-dependent) target distribution.

The degree of freedom of the chi-square distributions should also be carefully considered. The trivial mask pairs of the form $(\alpha, 0)$ or $(0, \beta)$ have constant zero correlation and reduce the effective degrees of freedom from the naive
$2^l - 1$ to $2^l - 2\cdot 2^n + 1$~\cite{KN19}. Approximations free of trivial components~\cite{Nyb19} recover the full $2^l - 1$ degrees of freedom. 

The multidimensional linear cryptanalysis framework branches in several directions in the subsequent works.
\emph{Zero\nobreakdash-correlation} linear cryptanalysis~%
\cite{bogdanov2014zero,bogdanov2012zerodata} exploits approximations with
$\mathrm{Cor}(B;\alpha,\beta)=0$ exactly, distinguishing $B$ from random
because the squared correlation of a random permutation is $\Theta(2^{-n})$ on
average. Bogdanov, Leander, Nyberg and Wang~\cite{bogdanov2012integral}, refined by
Sun et al.~\cite{sun2015links}, established a tight correspondence between
multidimensional zero\nobreakdash-correlation approximations with linearly
independent masks and integral distinguishers based on balanced functions, linking
two attack families that were previously studied in isolation.
\emph{Generalised} linear cryptanalysis on arbitrary finite abelian
groups~\cite{baigneres2007linear} replaces bit masks with group characters and
underlies, e.g., Beyne's analysis of the \textsf{FF3-1} and \textsf{FEA}
structures~\cite{beyne2021linear}.
 
In the quantum setting, Kaplan et al.~\cite{kaplan2016diff_lin} showed how quantum
counting~\cite{brassard2002amplitude} gives a quadratic speed-up for the
\emph{one-dimensional} variant: estimating $\#\{x\mid \alpha\cdot x=\beta\cdot
B(x)\}$ to the precision required for distinguishing takes time $O(1/|\mathrm{Cor}|)$
rather than $O(1/\mathrm{Cor}^{2})$. Lifting this speed-up to the multidimensional
regime is non-obvious, since the chi-squared statistic on $p_S^{B}$ is not the
population count of any natural Boolean function. Hosoyamada~%
\cite{hosoyamada2023multidim} bridged this gap with a minor modification of Simon's
subroutine---retaining the output register and Hadamard-transforming it---which
yields a state whose squared amplitude on $\ket{\alpha}\ket{\beta}$ equals
$\mathrm{Cor}(f;\alpha,\beta)^{2}/2^{n}$; combined with amplitude amplification,
this~\emph{correlation extraction algorithm} delivers speed-ups for multidimensional
linear, zero-correlation and balanced-property integral distinguishers, and on
occasion exceeds the quadratic barrier when several orthogonal integral properties
coexist.

\subsection{Differential-linear cryptanalysis}

Differential-linear cryptanalysis, introduced by Langford and Hellman~\cite{langford1994differential},
combines differential~\cite{biham1991differential} and linear~\cite{matsui1993linear}
cryptanalysis by splitting the cipher as $E = E_1 \circ E_0$, where $E_0$ has a
strong (truncated) differential and $E_1$ a biased linear approximation. For
an input difference $\delta$ and output mask $w$, the bias of the differential-linear approximation is
\begin{equation}
   \mathcal{E}_{\delta,w}
   := \Pr\!\bigl[\,w \cdot (E(x+\delta) + E(x)) = 0\,\bigr] - \tfrac{1}{2}.
\end{equation}

The classical estimate $\mathcal{E}_{\delta,w} \approx \varepsilon_{\delta,v}\, c_{v,w}^{2}$,
obtained via the Piling-up lemma~\cite{langford1994differential,biham2002enhancing},
combines the bias $\varepsilon_{\delta,v}$ of the truncated differential over $E_0$ with the
squared correlation $c_{v,w}^{2}$ of the linear approximation over $E_1$ where $v$ is the truncated differential output mask for the first half and input mask for the second half. Using the
Chabaud--Vaudenay link~\cite{chabaud1994links}, Blondeau, Leander, and
Nyberg~\cite{blondeau2014differential} showed that, assuming independence of $E_0$ and $E_1$,
\begin{equation}
   \mathcal{E}_{\delta,w}
   = \sum_{v \in \mathbb{F}_2^n} \varepsilon_{\delta,v}\, c_{v,w}^{2},
   \label{eq:dl-hull}
\end{equation}
the \emph{differential-linear hull} of $E$.

\paragraph{Multidimensional variant.}
For input differences in $U^{\perp}$ and output masks in $W$, the bias generalises to
\begin{equation}
   \mathcal{E}_{U,W}
   = \frac{2}{|W|} \sum_{v \neq 0} \varepsilon_{U,v}\, C_{v,W},
\end{equation}
where, $C_{v,W} = \sum_{w \in W\setminus\{0\}} c_{v,w}^{2}$ is the capacity of the
multidimensional linear approximation~\cite{blondeau2014differential}. Multiple input
differences enable structures that reduce data complexity by $|U^{\perp}|$, while
multidimensional output masks aggregate correlated approximations into a single capacity.

Bar-On et al. \cite{cryptoeprint:2019/256} proved that this independence assumption fails in many real-world cases and introduced a new tool, the Differential-Linear Connectivity Table (DLCT), to capture the dependency between the differential and linear parts at the boundary.  

\subsection{Fourier Fishing and quantum query lower bounds}
The hardness side of our results draws on a separate line of work concerning the
difficulty of \emph{finding} structure in a Boolean function, rather than
exploiting structure already known. Aaronson and
Chen~\cite{aaronson-chen}, in the course of laying complexity-theoretic
foundations for quantum supremacy, studied the \emph{Fourier Fishing} problem:
given oracle access to a Boolean function, output a point $z$ at which the
(squared, normalised) Walsh--Hadamard coefficient $|\widehat f(z)|^2$ is large.
They showed that, while a quantum algorithm solves this with a single query by
Hadamard-transforming the function and measuring, any classical randomized
algorithm needs $\Omega(N/\log N)$ queries to do appreciably better than chance.
This separation is the canonical example of a search problem that is easy
quantumly yet provably hard classically, and it sits within the broader
\emph{forrelation} family of problems known to exhibit maximal quantum--classical
query gaps. As we show in Section~\ref{sec:macfishing}, identifying mask pairs
with large masked cross-correlation is precisely a cross-correlation analogue of
Fourier Fishing; our classical lower bound (Theorem~\ref{thm:MAC_lb}) adapts the
Aaronson--Chen argument to this setting, and is to our knowledge the first such
hardness result for the mask-identification task underlying linear-style
distinguishers.

%% file: final/masked_auto_cor.tex
\section{Masked Auto-Correlation}\label{sec:mac}

In this section we establish the properties of masked auto-correlation used later
in the paper; several are of independent interest. We separate the \emph{exact}
identities, which hold for any permutation $f$, from the \emph{distributional}
behavior of $C_w(\alpha,\beta)$ for a random permutation. Throughout we abbreviate
$C^w_f(\alpha,\beta)$ to $C_w(\alpha,\beta)$ when $f$ is clear from context.

We first recall the central definitions formally; they were introduced
informally in Section~\ref{sec:overview}.

\begin{definition}[Masked cross-correlation]\label{def:mcc}
Let $f,g:\F_2^n\to\F_2^n$ be permutations, $\alpha,\beta\in\F_2^n$ output masks,
and $w\in\F_2^n$ a difference. The \emph{masked cross-correlation} of $f,g$ is
\[
  C^w_{f,g}(\alpha,\beta)
  = \stdnorm \sum_{x\in\F_2^n} (-1)^{\alpha\cdot f(x)}(-1)^{\beta\cdot g(x\xor w)}
  \;\in\; [-\sqrt N,\sqrt N].
\]
\end{definition}

\begin{definition}[Masked auto-correlation]\label{def:mac}
The \emph{masked auto-correlation} of a permutation $f$ is the case $g=f$ of
Definition~\ref{def:mcc},
\[
  C^w_f(\alpha,\beta)
  = \stdnorm \sum_{x\in\F_2^n} (-1)^{\alpha\cdot f(x)}(-1)^{\beta\cdot f(x\xor w)},
\]
written $C_w(\alpha,\beta)$ when $f$ is clear from context.
\end{definition}

\begin{definition}[Masked auto-correlation distribution]\label{def:macd}
For a fixed $w$, the \emph{masked auto-correlation distribution} $\CD_w$ over
mask pairs $(\alpha,\beta)\in\F_2^n\times\F_2^n$ assigns
\[
  \Pr_{\CD_w}[\alpha,\beta] = \tfrac{1}{N^2}\,C^w_f(\alpha,\beta)^2,
\]
which is a valid distribution by Lemma~\ref{lem:parseval-mac}.
\end{definition}
\subsection{Few identities}

We begin with a Parseval-type identity for masked auto-correlation.

\begin{lemma}\label{lem:parseval-mac}
    For any $w$, $\displaystyle\sum_{\alpha, \beta} C^{w}_f(\alpha,\beta)^2 = 2^{2n} = N^2$,
    and hence $\displaystyle\sum_{w, \alpha, \beta} C^{w}_f(\alpha,\beta)^2 = 2^{3n} = N^3$.
\end{lemma}

\begin{proof}
\begin{align*}
N\sum_{\alpha,\beta} C_f^{w}(\alpha, \beta)^2
  &= \sum_{\alpha,\beta}
     \Bigl(\sum_x (-1)^{\alpha \cdot f(x) \xor \beta \cdot f(x \xor w)}\Bigr)
     \Bigl(\sum_y (-1)^{\alpha \cdot f(y) \xor \beta \cdot f(y \xor w)}\Bigr) \\
  &= \sum_{x,y} \sum_{\alpha,\beta}
     (-1)^{\alpha \cdot (f(x)\xor f(y)) \,\xor\, \beta \cdot (f(x \xor w)\xor f(y\xor w))} \\
  &= \sum_{x=y} \sum_{\alpha,\beta}
     (-1)^{\alpha \cdot (f(x)\xor f(y)) \,\xor\, \beta \cdot (f(x \xor w)\xor f(y\xor w))} \\
  &\qquad
     + \sum_{x\neq y} \sum_{\alpha}(-1)^{\alpha \cdot (f(x)\xor f(y))}
       \sum_{\beta}(-1)^{\beta \cdot (f(x \xor w)\xor f(y\xor w))} \\
  &= N^3 + 0,
\end{align*}
since the $x=y$ terms each contribute $N^2$ over $N$ choices of $x$, while for
$x\neq y$ both $f(x)\xor f(y)$ and $f(x\xor w)\xor f(y\xor w)$ are non-zero strings
(as $f$ is a permutation), so each inner character sum over $\alpha$, resp.\
$\beta$, vanishes. Dividing by $N$ gives $N^2$.
\end{proof}

The next identity concerns zero-valued masks and follow from
$\sum_{x\in\F_2^n}(-1)^{a\cdot x}=0$ for any non-zero $a$.

\begin{lemma}\label{lem:zero-masks}
    For any $w$ and any non-zero $\alpha,\beta$,
    $C^w(0^n,0^n)=\sqrt{N}$, $\;C^w(0^n,\beta)=0$, and $C^w(\alpha,0^n)=0$.
\end{lemma}

Combining Lemmas~\ref{lem:parseval-mac} and~\ref{lem:zero-masks} gives
$\sum_{\substack{\alpha \not= 0^n\\\beta \not= 0^n}} C_w(\alpha,\beta)^2 = N^2 - N$, so there must
exist $(\alpha,\beta)\neq 0^{2n}$ with $C_w(\alpha,\beta)^2 \ge \tfrac{N-1}{N}$.
A useful consequence for the sampling distribution $\CD_w$ is the following.

\begin{corollary}\label{cor:support}
    The $\CD_w$ distribution never samples a pair $(\alpha,\beta)$ in which exactly
    one of $\alpha,\beta$ is the all-zeroes string: its support consists of pairs
    where either both are $0^n$, which happens with probability $\tfrac{1}{N}$, or both are non-zero, which happens with probability $\tfrac{N-1}{N}$.
\end{corollary}

\subsection{Distribution of MAC for a random permutation}

We now describe how $C_w(\alpha,\beta)$ behaves when $f=B$ is a random
permutation. The arguments are the standard central-limit heuristic, treating the
relevant masked outputs as sums of many near-independent $\pm1$ variables; we take
$N\to\infty$ so the limit applies. There are four cases, organised by whether
$w=0^n$ and whether $\alpha=\beta$.

\paragraph{Case $w=0^n$.}
If $\alpha=\beta$ then $C_w(\alpha,\beta)=\sqrt{N}$ identically. If
$\alpha\neq\beta$ then
\[
  C_w(\alpha,\beta) = \tfrac{1}{\sqrt{N}}\sum_x (-1)^{(\alpha\xor\beta)\cdot B(x)},
\]
and since $\alpha\xor\beta\neq 0^n$, the bit $(\alpha\xor\beta)\cdot B(x)$ behaves
as a uniform random Boolean function of $x$. The sum is therefore $N$
near-independent $\pm1$ terms, distributed as $\N(0,N)$, so after the $1/\sqrt N$
normalization $C_{0^n}(\alpha,\beta)\sim\N(0,1)$.

\begin{lemma}\label{lem:w0}
    For a random permutation $B$: $C_{0^n}(\alpha,\alpha)=\sqrt{N}$ always, and
    $C_{0^n}(\alpha,\beta) = 0,$ for distinct $\alpha,\beta$.
\end{lemma}

\paragraph{Case $w\neq 0^n$.}
The trivial subcase $\alpha=\beta=0^n$ again gives $C_w(0^n,0^n)=\sqrt{N}$. 

The two
non-trivial subcases differ in their variance. If $\alpha=\beta\neq 0^n$,
\[
  C_w(\alpha,\alpha) = \tfrac{1}{\sqrt{N}}\sum_x (-1)^{\alpha\cdot(B(x)\xor B(x\xor w))},
\]
and $B(x)\xor B(x\xor w)$ is a non-zero uniformly random string for each $x$. Pairing each
$x$ with $x\xor w$ leaves $N/2$ $\pm 1$ terms that can be assumed to be uniformly random and independent as $N \to \infty$, so $C_w(\alpha, \alpha)$ can be approximated as $\tfrac{2}{\sqrt{N}}$ times $\N(0,N/2)$; after the normalization, we get 
$C_w(\alpha,\alpha)\sim\N(0,2)$. 

If instead $\alpha\neq\beta$,
\[
  C_w(\alpha,\beta) = \tfrac{1}{\sqrt{N}}\sum_x (-1)^{\alpha\cdot B(x)\,\xor\,\beta\cdot B(x\xor w)},
\]
and for large $n$ the bits $\alpha\cdot B(x)$ and $\beta\cdot B(x\xor w)$ act as
independent uniform Boolean functions, giving $N$ near-independent $\pm1$ terms and
$C_w(\alpha,\beta)\sim\N(0,1)$. The two subcases are summarised as follows.

\begin{lemma}\label{lem:cd_w_normal_dist_case2}
    Let $w\neq 0^n$. Then $C_w(0^n,0^n)=\sqrt{N}$, $\;C_w(\alpha,\alpha)\sim\N(0,2)$,
    and $C_w(\alpha,\beta)\sim\N(0,1)$ for $\alpha\neq\beta$.
\end{lemma}

The above results are summarised in Table~\ref{tab:C_w_a_b}. 
\begin{table}[h]
\setlength{\arrayrulewidth}{0.5mm}
\setlength{\tabcolsep}{18pt}
\renewcommand{\arraystretch}{1.5}
    \centering
    \begin{tabular}{ccc}
                            & $w=0^n$ & $w \not= 0^n$ \\ [1ex]
         \hline \hline
    $\alpha = \beta = 0^n$  & $\sqrt{N}$ & $\sqrt{N}$ \\
    $\alpha \not= \beta = 0^n$ & 0 & 0 \\
    $0^n = \alpha \not= \beta$ & 0 & 0 \\ \hline
    $\alpha = \beta \not= 0^n$ & $\sqrt{N}$ & $\N(0,2)$ \\
    $0^n \not= \alpha \not= \beta \not= 0^n$ & $\N(0,1)$ & $\N(0,1)$ \\
    \end{tabular}
    \caption{Values of $C_w(\alpha,\beta)$. The last two rows  correspond to a random permutation.}
    \label{tab:C_w_a_b}
\end{table}

Next we present a few similar facts about $C_w(\alpha, \beta)^2$ which will be used in the hardness analysis in Theorem~\ref{thm:MAC_lb} (see
Appendix~\ref{sec:proof-MACfishing-lb}).

\begin{restatable}{lemma}{lemsimple}
\label{moments}
    Let $w, \alpha, \beta \neq 0^n$. Then:
    \begin{enumerate}
        \item For $\alpha \neq \beta$, $C_w(\alpha,\beta) \sim \N(0,1)$,
        $\mathbb{E}_B[C_w(\alpha,\beta)^2] = 1$, and $\Var_B[C_w(\alpha,\beta)^2]=2$.
        \item For $\alpha = \beta$, $C_w(\alpha,\beta) \sim \N(0,2)$,
        $\mathbb{E}_B[C_w(\alpha,\beta)^2] = 2$, and $\Var_B[C_w(\alpha,\beta)^2] = 8$.
        \item $\mathbb{E}_{\alpha,\beta}\mathbb{E}_B[C_w(\alpha,\beta)^2] = 1 + \tfrac{1}{N}$
        and
        $\Var_{\alpha,\beta,B}[C_w(\alpha,\beta)^2] = 2 + \tfrac{7}{N} - \tfrac{1}{N^2}$.
    \end{enumerate}
\end{restatable}

The proof essentially uses properties of normally distributed variables and is available in Appendix~\ref{app:lem1}.

Even though we considered $w=0^n$ in this section, we will {\em always} choose a non-zero shift $w$ for our algorithms and experiments.

%% file: final/mac_fishing.tex
\section{Masked Auto-Correlation (MAC) Fishing}\label{sec:macfishing}

Masked auto-correlation gives rise to a natural search problem: among the
exponentially many mask pairs $(\alpha,\beta)$, find one whose
$|C_w(\alpha,\beta)|$ is large. We call this \emph{MAC Fishing}. As we will see
in Section~\ref{sec:applications}, such heavy mask pairs are exactly what our
distinguishers use; but the problem is of independent interest, and we study
it here on its own terms. Throughout this section $B$ denotes the block cipher,
and we collect the heavy pairs into the set $V_w^\tau$ defined below, where
$\tau$ is an appropriately chosen threshold.

\begin{definition} \label{def:heavy_pairs}
    Define $V_w^\tau = \{ (\alpha, \beta) \in \F_2^n \times \F_2^n ~:~ C_w(\alpha, \beta)^2 \ge \tau \}$.
\end{definition}

Mask pairs in $V^\tau_w$, for a non-negligible $\tau$, may be useful towards distinguishing a cipher from a random permutation. The following lemma analyzes the number of mask pairs in $V^\tau_w$ for a random permutation.

\begin{restatable}{lemma}{vtaulemma}[Distribution of $|V_w^\tau|$]
\label{lem:vtaulemma}
Let $w\ne0^n$ and $B$ be a n-bit random permutation. For a fixed threshold $\tau>0,$ each pair of masks $(\alpha, \beta)$ belongs to the set $V^\tau_w$ with independent probability $p_{\alpha,\beta}$ given by
\begin{equation*}
    p_{\alpha,\beta}= 
    \begin{cases}
        2\Phi(-\sqrt{\tau}) 
        & \text{if } \alpha \neq \beta,\\
        2\Phi\!\left(-\sqrt{\tau/2}\right) 
        & \text{if } \alpha = \beta \neq 0^n.
    \end{cases}
\end{equation*}
    The expectation and variance of $|V_w^\tau|$ is given as
    \begin{align*}
        \mathbb{E}_B \left[|V_w^\tau|\right] &= N^2 \cdot p_1 + N \cdot (p_2-p_1), \\
         \text{Var}_B \left[|V_w^\tau|\right] &= (N^2-N) \cdot p_1(1-p_1) + N \cdot p_2(1-p_2)
    \end{align*}
    where, $p_1= 2\Phi(-\sqrt\tau)$ and $p_2= 2\Phi\left(-\sqrt{\frac{\tau}{2}}\right).$
\end{restatable}

The proof of the lemma is presented in Appendix~\ref{sec:vtaulemmasection}; 
the lemma claims the existence of a lot of good mask pairs for a random permutation. 
However, it is not so obvious that there would be any at all for a cipher given the rapid diffusion they exhibit. We show an evidence to the contrary in Figure~\ref{fig:v_tau}. Observe that lots of good mask pairs exist for a large number of keys in mini-AES irrespective of $w$; e.g., the left chart shows that at least $4\%$ pairs are good for nearly 75\% of the sampled keys. This provides a hope of success to our multidimensional approximation attack that requires many good mask pairs.


\begin{figure}[h]
    \centering
    \includegraphics[width=0.45\linewidth]{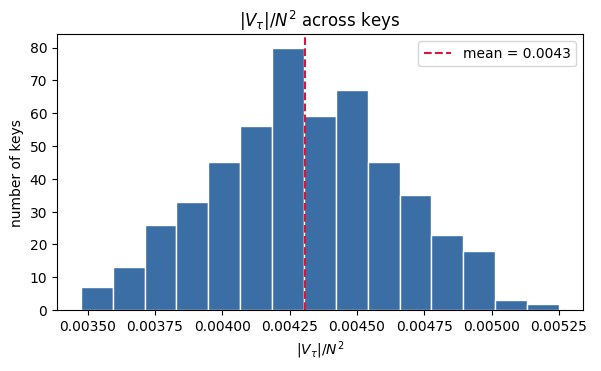} \hfill \includegraphics[width=0.45\linewidth]{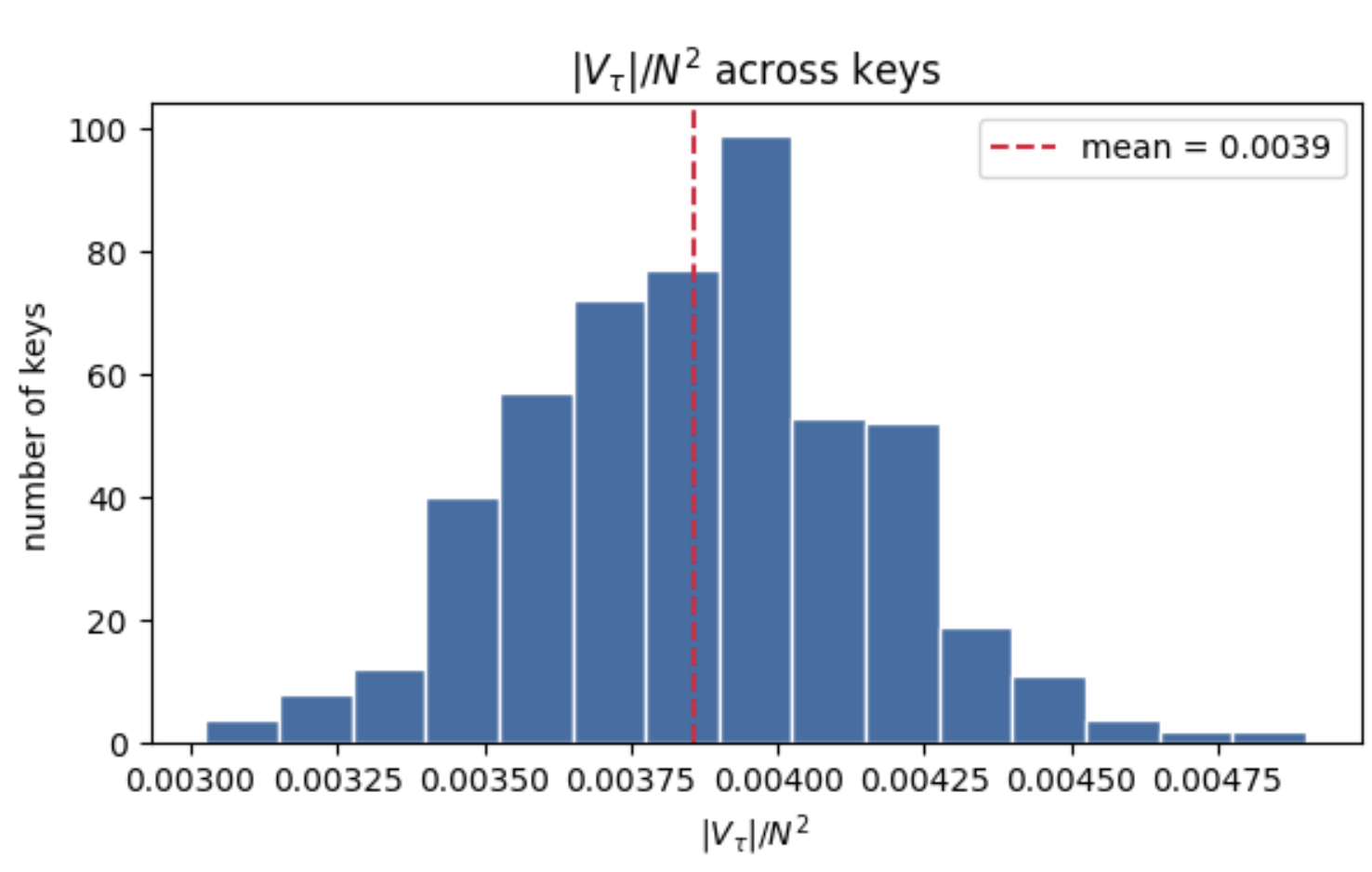}\\
    \includegraphics[width=0.45\linewidth]{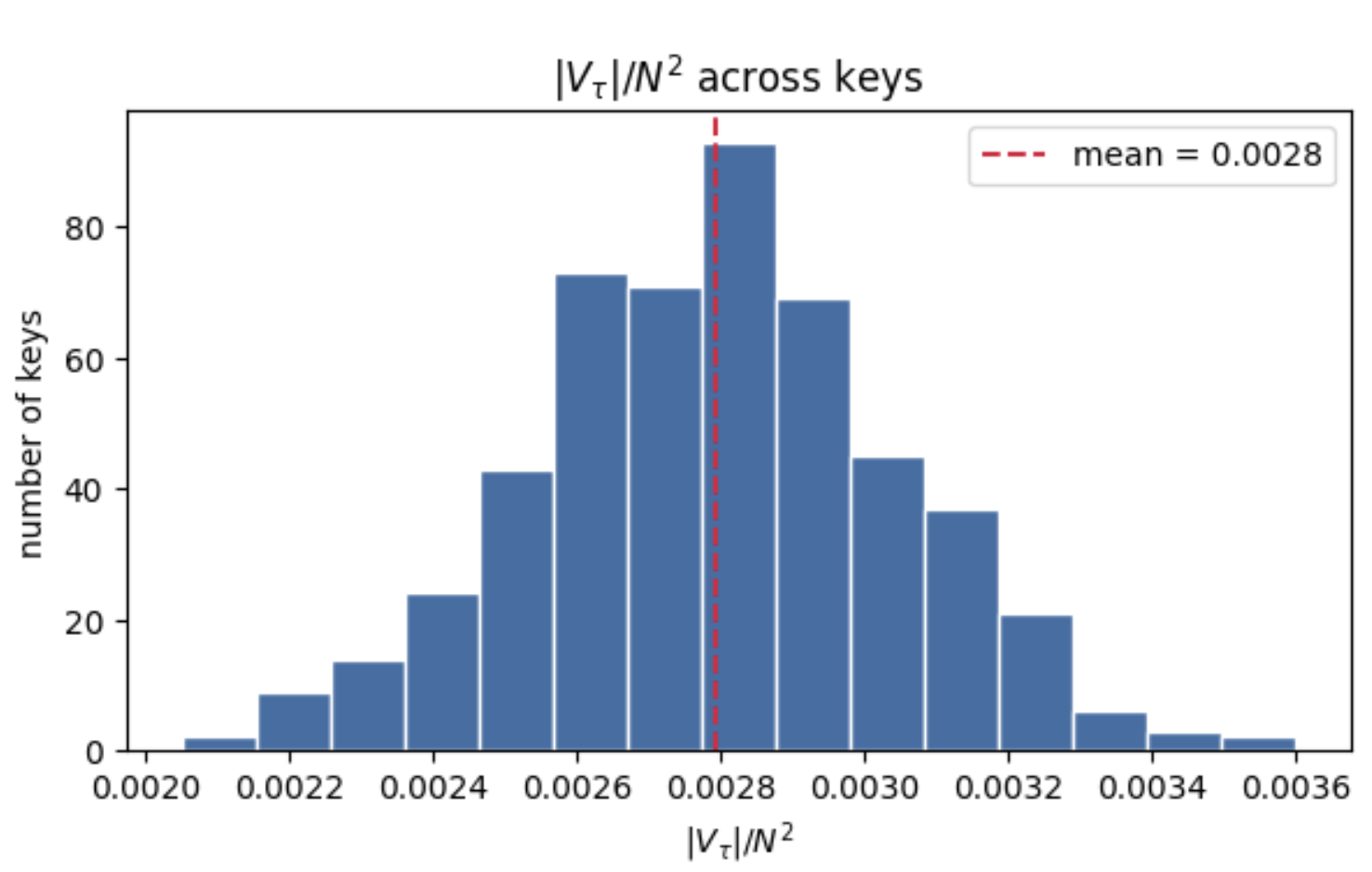} \hfill
    \includegraphics[width=0.45\linewidth]
    {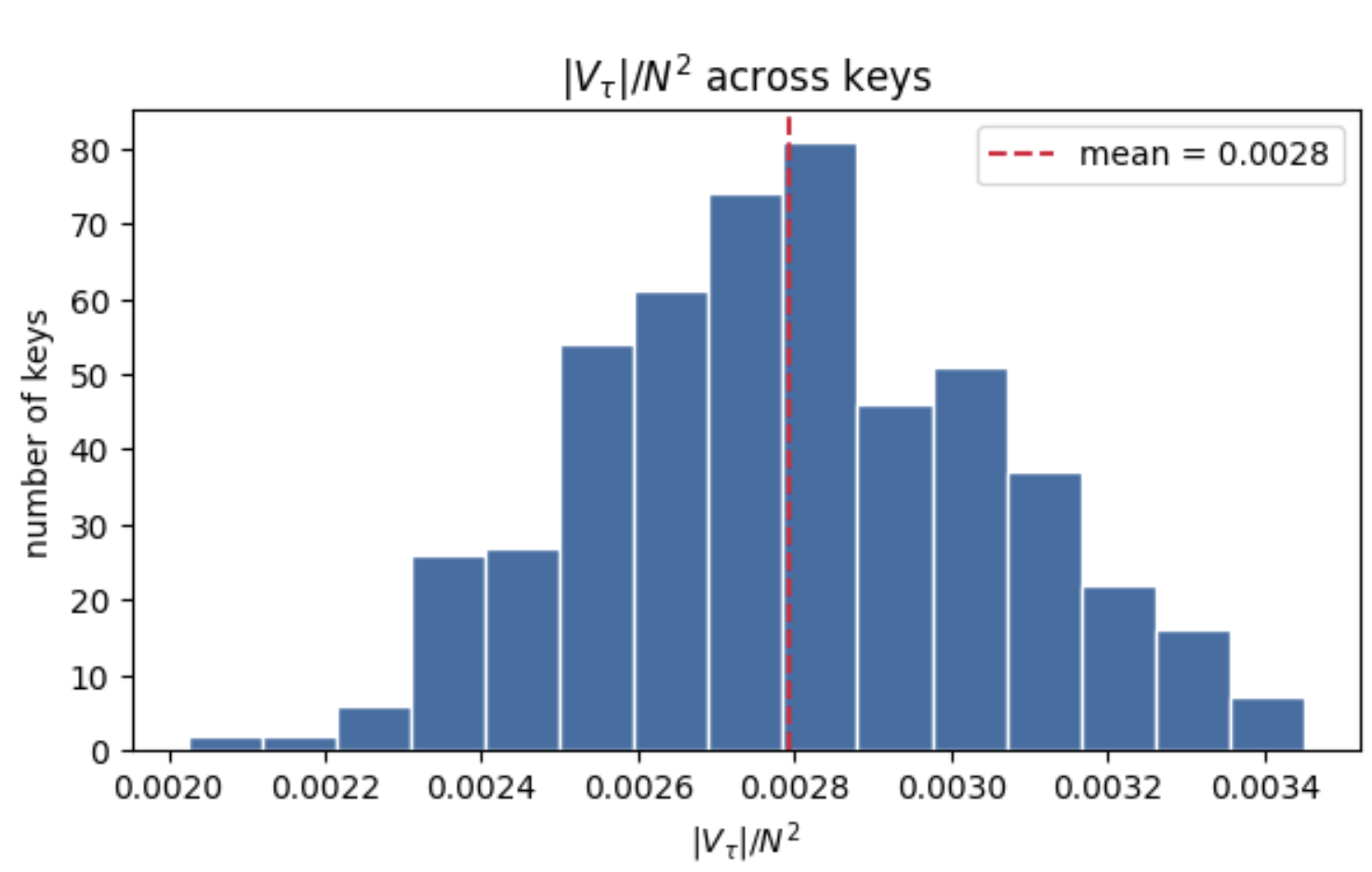}
    \caption{Distribution of the number of keys against the number of good mask pairs in 16-bit mini-AES, averaged over all keys. A total of 512 keys were randomly selected, sufficient number of mask pairs were sampled and the proportion of good pairs estimated to $0.95$ confidence (using Wilson score). $\tau$ was set to $9/2^{16}$. The values of $w$ that were chosen were 0x0003 (top-left), 0x3009 (top-right), 0x0845 (bottom-left), 0x2504 (bottom-right).}
    \label{fig:v_tau}
\end{figure}

We also show that MAC values have a wide-range of variations, even when considered across keys. E.g., the chart in Figure~\ref{fig:mini_AES_3round_example_masks} shows different $(\alpha, \beta, w)$ triplets with high, medium, and low values of MAC. This allows us to create $V^\tau_w$-based distinguishers with different levels of $\tau$.

\begin{figure}[h]
    \centering
    \includegraphics[width=0.8\linewidth]{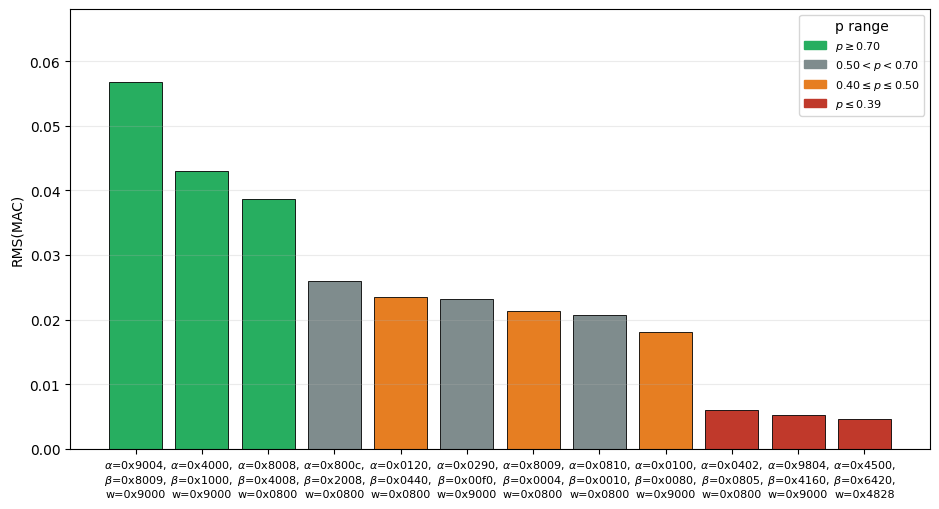}
    \caption{Variation in the values of MAC of sample $(\alpha,\beta,w)$ triplets across 50 random keys for 16-bit mini-AES. The root-mean-square (RMS) of the MAC values, across all keys, and with 95\% confidence (Wilson's score), are shown for all triplets.}
    \label{fig:mini_AES_3round_example_masks}
\end{figure}

Now we move to one of the central tasks towards constructing a MAC based distinguisher --- identifying good mask pairs.

\begin{definition}[Masked Auto-Correlation Fishing (MAC Fishing)]
    Given a constant $\tau$ and non-zero $w$, the task of the MAC Fishing problem is to output any $(\alpha,\beta) \in V^\tau_w$.
\end{definition}

We will be dealing with a promised version of the above problem --- $V^\tau_w$ is not empty.
Following the terminology of Aaronson and Chen~\cite{aaronson-chen} where they showed a lower bound on a similarly defined Fourier Fishing problem (on Walsh-Hadamard-Fourier coefficients), we define
$$\adv(w,\tau) = \tfrac{1}{N^2} \sum_{ (\alpha,\beta) \in V_w^\tau} C_w(\alpha, \beta)^2.$$

In Section~\ref{sec:MACfishing-qalgo} we present our constant-query Algorithm~\ref{algo:cd_w} that can solve the MAC Fishing problem with probability $\adv(w,\tau)$.

On the other hand, we were able to prove that the same task is computationally hard for a (randomized) classical algorithm, even when allowed exponentially more queries. 

\begin{theorem}\label{thm:MAC_lb}
Even given the promise that $\adv(w,1) \ge 0.801 - \tfrac{1}{n}$, any $o(N/\log N)$ classical randomized algorithm for MAC Fishing succeeds with probability at most $0.317 + o(1)$.
\end{theorem}

Using similar steps as was done for \cite[Theorem 9]{aaronson-chen}, we are able to prove the above claim whose proof is given later in Appendix~\ref{sec:proof-MACfishing-lb} (use $\tau=1$ to obtain the above theorem; similar separation claims can be easily obtained for any positive $\tau \ge 1$). The proof of the theorem also shows that there many functions that satisfy the promise of $\adv(w,1) \ge 0.801 - \tfrac{1}{n}$.


\subsection{Quantum Algorithm for MAC Fishing}\label{sec:MACfishing-qalgo}

We describe an efficient quantum sampling algorithm from $\CD_w$ in Algorithm~\ref{algo:cd_w}. Here we use the fact that $\ket{x} \mapsto \ket{B(x)}$ is a unitary mapping to save some qubits. 
We will later use this algorithm in our distinguisher.

\begin{algorithm}[h]

\caption{\textsc{MACSample}: sample from $\CD_w$}
\label{algo:cd_w}
\begin{algorithmic}[1]
\Require quantum oracle $U_B$ for the permutation $B$; difference $w\in\F_2^n$
\Ensure a mask pair $(\alpha,\beta)$ drawn with probability
  $|C_w(\alpha,\beta)|^2/N^2$
\State initialise three $n$-qubit registers as $\ket{w}\ket{0^n}\ket{0^n}$
\State apply $H^{\otimes n}$ to register 2
  \Comment{uniform superposition over $x\in\F_2^n$}
\State apply $\mathrm{CNOT}(\text{reg }2,\text{reg }3)$  \Comment{reg 3 holds $\ket{x}$}
\State apply $\mathrm{CNOT}(\text{reg }1,\text{reg }3)$  \Comment{reg 3 holds $\ket{x\oplus w}$}
\State apply $U_B$ to registers 2 and 3
  \Comment{$\ket{B(x)}\ket{B(x\oplus w)}$}
\State apply $H^{\otimes n}$ to registers 2 and 3
\State measure registers 2 and 3, obtaining $(\alpha,\beta)$
\State \textbf{if} exactly one of $\alpha,\beta$ equals $0^n$, or both do,
  \textbf{then} discard and repeat
\State \Return $(\alpha,\beta)$
\end{algorithmic}
\end{algorithm}

\paragraph{Correctness.}
On input $\ket{w}\ket{0^n}\ket{0^n}$, the registers evolve through
Algorithm~\ref{algo:cd_w} as
\begin{align*}
     \ket{w} \ket{0^n} \ket{0^n} 
\to & \stdnorm \sum_{x \in \mathbb{F}_2^n} \ket{w} \ket{x} \ket{0}
\to  \stdnorm \sum_{x \in \mathbb{F}_2^n} \ket{w} \ket{x} \ket{x \oplus w}
    \\
\to & \stdnorm \sum_{x \in \mathbb{F}_2^n} \ket{w} \ket{B(x)} \ket{B(x \oplus w)}
    \tag{apply $U_B$ to registers 2 and 3}
    \\
\to & \frac{1}{\sqrt{2^{3n}}} \sum_{x,\alpha,\beta \in \mathbb{F}_2^n}
    (-1)^{B(x) \cdot \alpha \oplus B(x \oplus w) \cdot \beta}
    \ket{w} \ket{\alpha} \ket{\beta}
    \\
\to & \tfrac{1}{N} \sum_{\alpha, \beta \in \mathbb{F}_2^n}
    C_w(\alpha, \beta) \ket{w} \ket{\alpha} \ket{\beta}.
\end{align*}

\begin{theorem}\label{lem:sample-prob}
Algorithm~\ref{algo:cd_w} makes at most three queries, in expectation, and outputs a non-zero pair of masks $(\alpha,\beta)$ with probability $C_w(\alpha,\beta)^2/N(N-1)$.

If line 8 is removed, it always makes two queries to $U_B$ and outputs a pair of masks $(\alpha,\beta)$ with probability $C_w(\alpha,\beta)^2/N^2$.
\end{theorem}

\begin{proof}
The final state above is the definition of $C_w(\alpha,\beta)$, so measuring the
last two registers samples $(\alpha,\beta)$ with probability
$|C_w(\alpha,\beta)|^2/N^2$, i.e.\ exactly from $\CD_w$. 

Line 8 is used to filter out the zero masks. By
Corollary~\ref{cor:support} the amplitude vanishes whenever exactly one of
$\alpha,\beta$ is $0^n$, so such pairs are never observed; the all-zero pair
$\alpha=\beta=0^n$ occurs with probability $1/N$ and is discarded by the repeat
step, after which the output is a pair of non-zero masks. The expected number of such repetitions is $N/(N-1) \le 1.5$.
\end{proof}

It is straightforward to extend the circuit to sample over all $w,\alpha,\beta$
jointly, by first preparing a uniform superposition over $\ket{w}$ and then
running the same steps, yielding
\[
  \frac{1}{N^{1.5}} \sum_{w,\alpha,\beta \in \mathbb{F}_2^n}
  C_w(\alpha, \beta)\,\ket{w}\ket{\alpha}\ket{\beta}.
\]

It is also possible to create a similar state but now including all possible keys in the superposition; the $U_B$ oracles now needs to act on the key register as well: $U_B \ket{k}\ket{x} = \ket{k}\ket{E_k(x)}$.
\[
  \frac{1}{N^{2}} \sum_{k,w,\alpha,\beta \in \mathbb{F}_2^n}
  C_w(\alpha, \beta)\,\ket{k} \ket{w}\ket{\alpha}\ket{\beta}.
\]

Doing so, and measuring the $w$-register too, will allow us to sample for $(w,\alpha,\beta)$ such that $|C_w(\alpha, \beta)|$ is high across all keys.

%% file: final/mdl_approx.tex
\section{Masked Differential-Linear (MDL) Approximation}\label{sec:mdl}

Let's formally introduce the notion of MDL approximations of the block cipher $B$ which is designed similar to linear approximations of vectorial functions.

\begin{definition}[Masked Differential Linear (MDL) Approximation]
    A (one-dimensional) masked-differential-linear approximation is the map
    $$g_w^{\alpha,\beta}:=(x) \mapsto \alpha \cdot f(x) \xor \beta \cdot f(x \xor w)$$
    where $x \in \F_2^n$ is an input, $\alpha, \beta \in \F_2^n$ are output masks and $w \in \F_2^n$ is an input-shift/difference.
\end{definition}

The correlation, $Cor(f,w;\alpha, \beta)$, of such an approximation is defined as
$$Cor(f,w;\alpha, \beta) = \Pr_x[\alpha \cdot f(x) = \beta \cdot f(x \xor w)]- \Pr_x[\alpha \cdot f(x) \not= \beta \cdot f(x \xor w)].$$

Of course, $Cor(f,w; \alpha, \beta) = \tfrac{1}{N}\sum_x (-1)^{g_w^{\alpha,\beta}(x)} = C^w_f(\alpha,\beta)/\sqrt{N}$.


\subsection{Single Approximation}\label{subsec:single-approx}
Let's discuss how to distinguish between a cipher and a random permutation using $g_w^{\alpha,\beta}$ for some pair of non-zero masks $(\alpha,\beta) \in \F_2^n \times \F_2^n$. We will assume at this point that we know the approximating masks $(\alpha,\beta)$ and a lower bound $\epsilon$ on the approximations, i.e., say, $|\cor(f,w;\alpha,\beta) |\ge \epsilon$. Recall that our MAC Fishing algorithm was designed to provide these information.

Observe that we can write $C^{w}_f(\alpha,\beta)/\sqrt{N}$ as

$$\E_x \left[ (-1)^{\alpha\cdot f(x) \oplus \beta \cdot f(x \oplus w)} \right] = 1 - 2 \E_x \left[ \alpha\cdot f(x) \oplus \beta \cdot f(x \oplus w) \right],$$
so $|\cor(f,w;\alpha,\beta)| \ge \epsilon$ implies that 
$$\E_x[\alpha \cdot f(x) \xor \beta \cdot f(x \xor w)] \le (1-\epsilon)/2 \text { or } \ge (1+\epsilon)/2.$$

On the other hand, if $f$ is a random permutation, then it is not difficult to show that $\E_x[\alpha \cdot f(x) \xor \beta \cdot f(x \xor w)] = 1/2$.

To launch an attack, randomly choose a set $S$ plaintexts, and then compute the average of $\alpha \cdot f(x) \xor \beta \cdot f(x \xor w)$ for $x \in S$, denoted $\widehat{\cor}$. If $\widehat{\cor} \le \epsilon/2$ then the data can be assumed to come from a random permutation, else it can be predicted to have come from a block cipher and high-correlation approximation.
It can shown using Chernoff's bound that $O(1/\epsilon^2)$ classical samples are sufficient for guessing with exponentially high confidence. Using the standard quantum algorithm for counting, the number of quantum queries (to the block cipher oracle) can be reduced to $O(1/\epsilon)$~\cite[Sec.~7.1.1,7.2.1]{kaplan2016diff_lin}.

\subsection{Multidimensional Approximation}

Let $g^w$ denote a collection of (single) approximations. We will formally explain this below, but when we have to deal with multiple approximations, the idea is to focus on the value distribution $\Pr_{z \in \F_2^n} ~[g^w(x) = z]$ of linear-combinations of the approximations, denoted $g^w$, instead of $f$~\footnote{Value distribution of $f$ is always uniform since $f$ represents a permutation.}. Let $p(z)$ denote this distribution. Following existing approaches (see Section~\ref{subsec:multidim-lin-crypt}) and making use of Lemma~\ref{lem:zero-masks}, the capacity of this distribution can be determined from the correlations of $g^w$. Now we formalize this notion.

\begin{definition}[Multi-dimensional MDL Approximation]\label{defn:MDL}
    Fix a difference $w \in \mathbb{F}_2^n$ and a collection of $t$ non-zero linearly-independent mask pairs $M=(\alpha_i, \beta_i) \in \mathbb{F}_2^n \times \mathbb{F}_2^n$ for $i = 1,\dots,t$. 
    A multi-dimensional masked-differential-linear approximation is a $t$-bit function $g^w:\F_2^n \to F_2^t$ that is constructed by stacking the $t$ masked differential bits,
\begin{equation}\label{eq:gw-explicit}
  g^w(x) \;=\; \bigl(
    \alpha_1 \cdot f(x) \oplus \beta_1 \cdot f(x \oplus w),\;
    \dots,\;
    \alpha_t \cdot f(x) \oplus \beta_t \cdot f(x \oplus w)
  \bigr) \;\in\; \mathbb{F}_2^t.
\end{equation}
\end{definition}

So, each output coordinate corresponds to one mask pair $(\alpha_i,\beta_i)$. Equivalently, $g^w(x) = (\alpha_i,\beta_i) \cdot (f(x), f(x\oplus w))$ which is essentially a linear projection of $(f(x),f(x\xor w))$ determined by the masks.

The next lemma is well known in the literature; it connects the capacity of an MDL to the value distribution of $g^w$ and the sum of masked auto-correlations for a special set of mask pairs. It relies on a crucial property of this construction that every linear combination of the
coordinates of $g^w$ is again a single masked differential approximation of $f$.

\begin{lemma}\label{lemma:boxed}
Capacity of a multi-dimensional approximation can be expressed as the following equivalent forms. Here, $p_z$ denotes $\Pr_x[g^w(x)=z]$.
$$
  \mathrm{Cap}(g^w)
  \;=\;
  2^t \sum_{z\in\mathbb{F}_2^t}\Bigl(p_z - 2^{-t}\Bigr)^2
  \;=\;
  \sum_{S\neq 0}\bigl(\mathrm{Cor}(S\cdot g^w)\bigr)^2
  \;=\;
  \tfrac{1}{N}\sum_{S\neq 0}\bigl(C_f^w(\alpha_S,\beta_S)\bigr)^2.
$$
\end{lemma}

\begin{proof}
For any nonzero $S = (S_1,\dots,S_t) \in \mathbb{F}_2^t$,
\begin{equation}\label{eq:linear-combo}
  S \cdot g^w(x)
  \;=\;
  \Bigl(\sum_{i} S_i \alpha_i\Bigr) \cdot f(x)
  \;\oplus\;
  \Bigl(\sum_{i} S_i \beta_i\Bigr) \cdot f(x \oplus w)
  \;=\; \alpha_S \cdot f(x) \oplus \beta_S \cdot f(x\oplus w),
\end{equation}
with $\alpha_S = \sum_i S_i\alpha_i$ and $\beta_S = \sum_i S_i\beta_i$, so that
$\mathrm{Cor}(S\cdot g^w) = C_f^w(\alpha_S,\beta_S)$, the
$(\alpha_S,\beta_S,w)$-masked auto-correlation of $f$. By the Parseval identity
for $\mathbb{F}_2^t$-valued functions, the capacity of $g^w$ therefore decomposes
exactly into the masked cross-correlations of $f$.
\end{proof}

The set of masks used to calculate capacity consists of all possible non-empty linear combinations of $[(\alpha_1, \beta_1), \ldots, (\alpha_t, \beta_t)]$. Thus, to ensure that every combination is used exactly once, the original set of masks should be linearly independent. That was the reason for requiring the masks to be linearly independent in Definition~\ref{defn:MDL}.




\subsubsection{MDL Capacity}

Let's discuss how the capacity of an MDL approximation behaves for a random permutation. We follow the approach of Khan and Nyberg~\cite{KN19}.

Let $L\subset \mathbb{F}_2^n \times \mathbb{F}_2^n$ denote the $t$-dimensional vector subspace for our masks. Without loss of generality we can assume that $L$ does not contain any $(\alpha,\beta)$ with non-zero $\alpha$ or non-zero $\beta$ (Lemma~\ref{lem:zero-masks}). Further, $C^w_f(\alpha, \beta)$ is 0, on expectation, for a random $f$ and any pair of masks; thus, $\Cap(L)=0$, implying that $p_z = \tfrac{1}{2^t}$. Thus, the capacity of $L$ follows a central $\chi^2$-distribution with at most $2^t-1$ degrees of freedom.






However, we believe that the following conjecture is true, similar to one proposed by Ashur et al.~\cite{Ashur2022,hosoyamada2023multidim}.

\begin{conjecture}\label{conj:capacity}
    Fix any $w \not= 0^n$. Let $L \subseteq \F_2^n \times \F_2^n$ be a vector space of non-zero mask pairs. 
    For a random block cipher $B$, its capacity approximately follows the central $\chi^2$ distribution with $2^t - 1$ degrees of freedom.
\end{conjecture}

The above conjecture should hold in case all the mask pairs in $L$ are non-zero; indeed, that holds true for the masks returned by our \textsc{MACSample} algorithm (Algorithm~\ref{algo:cd_w}). In case there are zero masks in $L$, the degree of freedom can be appropriately reduced~\cite[Conjecture~1]{Ashur2022}. We also observed the conjecture to hold in our experiments, e.g., see Figure~\ref{fig:capacity_dist}.

\begin{figure}[h]
    \centering
    \includegraphics[width=0.75\linewidth]{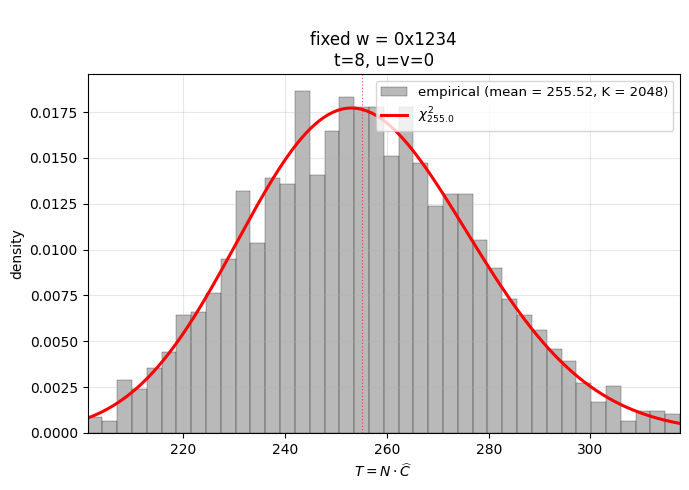}
    \caption{Distribution of capacity for 16-bit random permutations follow $\chi^2(255)$ for an 8-dimensional vector space of mask pairs; $u,v$ denotes the dimensions of zero-masks which were not used in this experiment. $X$-axis indicates capacity (scaled by $N$) and density indicate the probability density function.}
    \label{fig:capacity_dist}
\end{figure}

\subsection{Relationship to Existing Techniques}
\label{sec:mdl-relations}
It is important to note that our technique, for different classes of linear attacks, is a strict generalisation. It successfully captures not only the structure for these cryptanalytic attacks but even shows that the same statistical properties hold within the generalisation. 

\subsubsection{Multidimensional Linear Cryptanalysis}

Classical linear cryptanalysis studies approximations of the form
\(
\alpha \cdot x \oplus \beta \cdot f(x),
\)
whose bias is measured through the correlation between input and output masks. In our setting, if we fix one of the functions to be the identity map and choose $w=0$, i.e.,
\[
f(x)=x, \qquad g(x)=f(x), \qquad w=0,
\]
then the MDL approximation
\(
\alpha \cdot f(x) \oplus \beta \cdot g(x \oplus w)
\)
reduces exactly to the standard linear approximation. Consequently, the associated multidimensional setting, where masks form a vector space and the distinguisher is based on the sum of squared correlations (capacity), is recovered naturally within our framework. Thus, multidimensional linear cryptanalysis corresponds to the special case where no input difference is introduced and one side of the approximation acts directly on the plaintext.

\subsubsection{Differential-Linear Cryptanalysis}

Differential-linear cryptanalysis studies the propagation of an input difference through a cipher followed by a linear approximation on the output. In its standard form, for an input difference $\delta$ and output mask $\gamma$, one considers the bias of
\(
\gamma \cdot \bigl(f(x) \oplus f(x \oplus \delta)\bigr).
\)
Our framework generalizes this viewpoint by allowing \emph{two independent output masks} across the shifted evaluations of the cipher:
\(
\alpha \cdot f(x) \oplus \beta \cdot f(x \oplus w).
\)

The classical differential-linear setting is recovered by taking $\alpha=\beta=\gamma$ and $w=\delta$. Hence, differential-linear cryptanalysis appears as a restricted subclass of MDL approximations in which identical masks are enforced across the two correlated branches. By permitting arbitrary mask pairs $(\alpha,\beta)$ and organizing them into vector spaces, our framework further extends naturally to a multidimensional differential-linear setting.

\subsubsection{Differential-Linear Connectivity Table (DLCT)}
The Differential-Linear Connectivity Table (DLCT) can be viewed as a special case of the proposed MDL framework. Recall that for a vectorial Boolean function $S : \mathbb{F}_2^n \to \mathbb{F}_2^n$, the DLCT entry corresponding to an input difference $\Delta$ and output mask $\lambda$ is defined as
\[
\mathrm{DLCT}_S(\Delta,\lambda)
=
\left|\{x \in \mathbb{F}_2^n :
\lambda \cdot S(x)
=
\lambda \cdot S(x \oplus \Delta)\}\right|
-2^{n-1}.
\]
Now consider the masked cross-correlation in our framework with $f=g=S$, $\alpha=\beta=\lambda$, and shift $w=\Delta$. Then,
\[
\cor(S,\Delta; \lambda, \lambda)
=
\frac{1}{\sqrt{N}}
\sum_{x \in \mathbb{F}_2^n}
(-1)^{\lambda \cdot S(x)}
(-1)^{\lambda \cdot S(x \oplus \Delta)},
\]
which measures the correlation between two boolean functions by masking the output. By separating the summation into agreeing and disagreeing terms, one obtains that the DLCT entry is proportional to this differential-linear correlation value, up to normalization. Hence, the DLCT naturally emerges as a restricted instance of our MDL framework where the same function and identical masks are used. In this sense, MDL approximations strictly generalise differential-linear cryptanalysis and the DLCT formalism.

%% file: final/experiments.tex
\newcommand{\phat}{\hat{p}}
\newcommand{\rhohat}{\hat{\rho}}

%% file: final/attacks.tex
\section{Distinguishers and Key-Recovery Attacks}\label{sec:applications}

The two preceding sections give us, respectively, a way to find high-correlation
mask pairs (MAC Fishing) and a framework in which those pairs define
multidimensional approximations (MDL). We now turn these into concrete
cryptanalysis. We present two families of attacks. The first is a capacity-based
attack whose processing is classical once masks are supplied: given a set of mask
pairs it distinguishes a cipher from a random permutation and recovers a
last-round subkey (Section~\ref{sec:cap-attacks}). The second obtains good masks
with the quantum MAC Fishing primitive, yielding a distinguisher and key-recovery
attack whose mask-finding step has no efficient classical analogue
(Section~\ref{sec:mac-attacks}). In both families the masks originate from the
same quantum sampler (Algorithm~\ref{algo:cd_w}); the MAC-Fishing family additionally verifies each
correlation quantumly. We also validate our results on a toy Simon cipher, the results described in Appendix~\ref{sec:exp-validation}.

\paragraph{Correlation convention.}
Throughout this section we measure approximations by the normalised correlation
$\cor(B,w;\alpha,\beta)=C^w_B(\alpha,\beta)/\sqrt N\in[-1,1]$ of
Section~\ref{sec:mdl}, so that for a random permutation
$\cor(B,w;\alpha,\beta)\sim\N(0,2^{-n})$ and $\E[\cor^2]=2^{-n}$. Capacities and
decision thresholds are stated in these units. Thresholding
$\cor(B,w;\alpha,\beta)^2\ge\tau$ selects the heavy set $V_w^{N\tau}$ of
Section~\ref{sec:macfishing} (defined there in the Walsh scaling $C^w_B$), and a
single draw of \textsc{MACSample} lands in it with probability
$\mathrm{adv}:=\mathrm{adv}(w,N\tau)$. The distinguishers below use
$\tau=\omega(\log N/N)$, for which a random permutation has no heavy pair with high
probability (Section~\ref{sec:mac-dist}); this is sparser than the
constant-threshold regime $N\tau=\Theta(1)$ in which the hardness theorem is
proved.

\subsection{A Capacity-Based Block Cipher Distinguisher and Key Recovery}
\label{sec:cap-attacks}

\subsubsection{Capacity Based Distinguisher Algorithm}
The distinguisher is built upon Lemma~\ref{lemma:boxed} that expresses $\Cap(g^w)$ as sum of squares of masked auto-correlations for a set of carefully created masks. 
\begin{equation}
  \mathrm{Cap}(g^w)
  \;=\;
  \sum_{S \neq 0}
  \cor(B,w;\alpha_S,\beta_S)^2 .
\end{equation}

It is well known that $\cor(B,w;\alpha, \beta)$ is identical to the Fourier coefficient of $\Vec{p}$ (value distribution of $g^w$) at $(\alpha,\beta)$; thus, capacity measures how far that distribution deviates from uniform: a random
permutation has $C_B^{\alpha,\beta}(w)\approx 0$ for all masks, hence
$\mathrm{Cap}(g^w) \approx 0$ (Conjecture~\ref{conj:capacity}), whereas a structured cipher accumulates nonzero correlations across the mask pairs and pushes $\mathrm{Cap}(g^w)$ above the random baseline.


Algorithm~\ref{alg:cap-dist} formalizes the distinguisher. It estimates
$\mathrm{Cap}(g^w)$ empirically and thresholds it: given $N$ plaintext pairs under
a fixed difference $w$, it estimates the capacity of $g^w$ and compares it against
a threshold $\tau$ to decide whether the oracle behaves as a structured cipher or
a random permutation.

\begin{algorithm}[h]
\caption{\textsc{CapacityDistinguisher}}
\label{alg:cap-dist}
\begin{algorithmic}[1]
\Require permutation $B:\mathbb{F}_2^n\to\mathbb{F}_2^n$; difference $w$; mask
  pairs $(\alpha_i,\beta_i)_{i=1}^t$; sample count $C$; threshold $\tau$
\Return ``cipher'' or ``random permutation''
\State Initialize $\widehat{\mathrm{Cap}}(g^w) \gets 0$
\For{$s \in \F_2^t \setminus \{0^t\}$}
    \State $\alpha_s \gets \sum_{i=1}^t s_i \alpha_i$
    \State $\beta_s \gets \sum_{i=1}^t s_i \beta_i$
    \State Choose $C$ many $x \in \F_2^n$ uniformly at random
    \State Use the sampled $x$ to estimate $c \gets 1-2\Pr[\alpha_s \cdot B(x) \xor \beta_s \cdot B(x \xor w)]$
    \State Add $c^2$ to $\widehat{\mathrm{Cap}}(g^w)$
\EndFor
\If{$\widehat{\mathrm{Cap}}(g^w) > \tau$}
  \State \Return ``cipher''
\Else
  \State \Return ``random permutation''
\EndIf
\end{algorithmic}
\end{algorithm}





The estimation of the probability in Line 6 of the algorithm can be done using both classical techniques, requiring $C=O(1/\epsilon^2)$, or quantum techniques, requiring $C=O(1/\epsilon)$ queries to an oracle for $B$.

\begin{figure}[h]
    \centering
    \includegraphics[width=\linewidth]{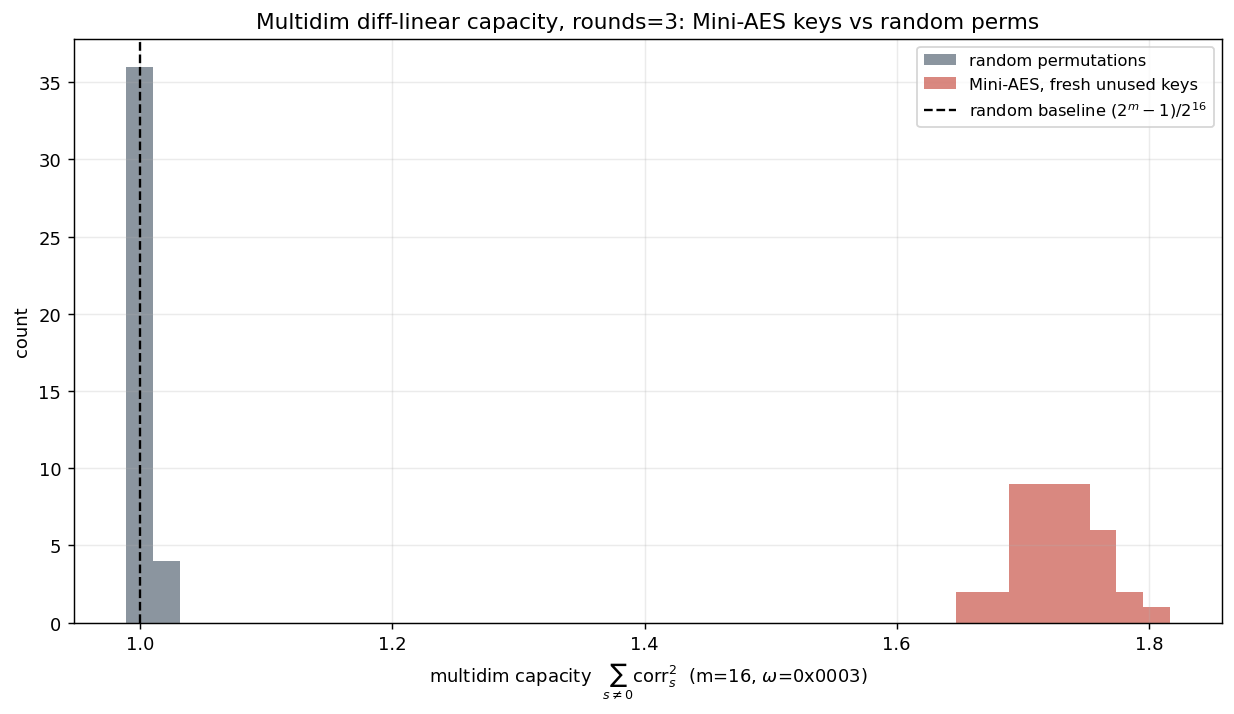}
    \caption{Distributions of keys that are vulnerable to capacity based attack. The 16-dimensional multidimensional approximation was chosen earlier by identifying good mask pairs across a set of randomly chosen keys. A fresh set of 40 keys was used for this experiment to prevent any bias arising from using the same keys for attacking as that for identifying the masks. $X$-axis shows the capacity.}
    \label{fig:Capacity distinguisher Mini-AES}
\end{figure}

We validate the effectiveness of \textsc{CapacityDistinguisher} on Mini-AES. 
In Figure~\ref{fig:Capacity distinguisher Mini-AES} we show how capacity distinguisher can identify between a random 16-bit permutation and Mini-AES for the following set of masks and $w=0x0003$. 

\begin{tabular}{cc|cc|cc|cc}
\hline
$\alpha$ & $\beta$ & $\alpha$ & $\beta$ & $\alpha$ & $\beta$ & $\alpha$ & $\beta$ \\
\hline
0x0060 & 0x0080 & 0x3000 & 0x6000 & 0x0080 & 0x0060 & 0x0003 & 0x0007 \\
0x0250 & 0x0890 & 0x2000 & 0x1008 & 0x1002 & 0xc003 & 0x0020 & 0x0050 \\
0x0060 & 0x0050 & 0x0100 & 0x0200 & 0xa000 & 0x3000 & 0x0050 & 0x0020 \\
0x0410 & 0x04a0 & 0x0200 & 0x0100 & 0x0120 & 0x0890 & 0xe000 & 0x9000 \\
\hline
\end{tabular}

Almost all keys led to MLD capacity of nearly 1 for a random permutation but close to 1.7 for Mini-AES. Further, the use of Wilson's lower bound for confidence lets us confidently extend the inference across all keys.


An entire pipeline of key-recovery using the capacity-based distinguisher on Simon is given in Appendix~\ref{sec:exp-capacity}.

\subsubsection{Key Recovery on the $(r+1)$-Round Cipher}\label{sec:cap-kr}

The distinguisher above, effective against the $r$-round reduced cipher $B$,
extends to a last-round key-recovery on the $(r+1)$-round cipher $E_K$ in the
standard way. Write $E_K = R_{k^\star}\circ B$, where $R_{k^\star}$ is the final
round under the (unknown) last-round subkey $k^\star$ and $R_k^{-1}$ is its
inverse under a guess $k$. For each candidate $k$ define the partially decrypted
oracle
\[
  B_k(x) \;=\; R_k^{-1}\!\bigl(E_K(x)\bigr),
\]
so that $B_{k^\star}=B$ recovers the $r$-round structure, while for $k\neq
k^\star$ the extra mis-decrypted round randomizes $B_k$ (wrong-key randomization
hypothesis): its masked correlations sit at the random baseline,
$\cor(B_k,w;\alpha,\beta)\sim\N(0,2^{-n})$.

Let $\{(w_i,\alpha_i,\beta_i)\}_{i=1}^m$ be the heavy tuples produced by
\textsc{MACSample} (Algorithm~\ref{algo:cd_w}) on the reduced cipher --- the same
draws that drive the distinguisher --- and abbreviate
$\cor_i(k)=\cor(B_k,w_i;\alpha_i,\beta_i)$. For each $k$ we estimate the signed
correlation vector $v(k)=\bigl(\cor_1(k),\dots,\cor_m(k)\bigr)\in[-1,1]^m$.

\paragraph{Correlation-filter.}
The unsigned statistic is the average squared correlation over the template
tuples,
\[
  \widehat{\mathrm{adv}}(k)
  \;=\; \frac{1}{m}\sum_{i=1}^m \cor_i(k)^2
  \;=\; \frac{1}{m}\,\|v(k)\|^2 ,
\]
the empirical analogue of the fishing advantage restricted to the template. Note
this is the energy over the $m$ chosen tuples. Under wrong-key randomisation $\E[\widehat{\mathrm{adv}}(k)]\approx
2^{-n}$ for $k\neq k^\star$, whereas
$\widehat{\mathrm{adv}}(k^\star)=\tfrac1m\sum_i\cor_i^{\star2}\ge\tau$ since each
template pair is heavy ($\cor_i^{\star2}\ge\tau$). Filtering using $\widehat{\mathrm{adv}}(k)\ge\tau$ therefore removes the bulk of wrong guesses.

\paragraph{Sign degeneracy.}
Screening by $\widehat{\mathrm{adv}}$ alone does not isolate $k^\star$. When the
last-round subkey enters the masked output bits linearly --- the generic case ---
each component correlation factors as
\[
  \cor_i(k) \;\approx\; (-1)^{\langle k,\,m_i\rangle}\,\cor_i^\star ,
\]
where $m_i$ is the key mask induced by $(\alpha_i,\beta_i)$ and the round
structure and $\cor_i^\star=\cor_i(k^\star)$. The magnitudes $|\cor_i(k)|$ are then
\emph{key-independent}, so every $k$ in the linear class shares $k^\star$'s energy
and survives the screen as an equivalence class. What distinguishes $k^\star$
within the class is not the magnitudes but the signs $(-1)^{\langle k,m_i\rangle}$.

\paragraph{Signed template matching.}
We resolve the class by comparing the full signed vector against a reference
\emph{template} $v^\star=(\cor_1^\star,\dots,\cor_m^\star)$ via cosine similarity,
\begin{equation}\label{eq:cosine}
  \mathrm{Score}(k)
  = \frac{\langle v(k),v^\star\rangle}{\|v(k)\|\,\|v^\star\|},
  \qquad
  \langle v(k),v^\star\rangle = \sum_{i=1}^m \cor_i(k)\,\cor_i^\star ,
\end{equation}
which is maximized exactly when the signs of $v(k)$ align with the template, i.e.\
at $k=k^\star$. The reference $v^\star$ is the signed profile of the $r$-round
reduced cipher; since these signs depend on the inner-round keys, we obtain
$v^\star$ in the offline phase from the reduced-cipher oracle --- built from the
same \textsc{MACSample} draws that fix the template tuples --- which is the
known-key setting of our experiments. A fully online recovery of the sign pattern,
without the reduced-cipher reference, we leave to future work.

\paragraph{The attack.}
Algorithm~\ref{alg:template-key-recovery} combines the two stages: a filter discards candidates whose correlation energy is at the random baseline,
and the survivors $\mathcal{K}'$ are re-ranked by the signed
score~\eqref{eq:cosine}. In our experiments this recovered the correct subkey at
the top of the ranking, where the energy screen alone returned a whole
equivalence class.

\begin{algorithm}[h]
\caption{\textsc{TemplateKeyRecovery}}
\label{alg:template-key-recovery}
\begin{algorithmic}[1]
\Require $(r+1)$-round cipher $E_K$; heavy tuples $\{(w_i,\alpha_i,\beta_i)\}_{i=1}^m$
  from \textsc{MACSample}; candidate subkeys $\mathcal{K}$; reference template
  $v^\star=(\cor_1^\star,\dots,\cor_m^\star)$; threshold $\tau$; sample count $N$
\Ensure ranked list of last-round subkey candidates
\Statex \textbf{Phase 1: score every candidate key}
\For{each $k\in\mathcal{K}$}
  \State $v(k)\gets ()$
  \For{$i = 1$ \textbf{to} $m$}
    \State estimate $\cor_i(k)=\cor(B_k,w_i;\alpha_i,\beta_i)$ from $N$ pairs,
      with $B_k(x)=R_k^{-1}(E_K(x))$; append to $v(k)$
  \EndFor
  \State $\widehat{\mathrm{adv}}(k)\gets\tfrac1m\sum_i \cor_i(k)^2$;\quad
    $\mathrm{Score}(k)\gets\langle v(k),v^\star\rangle/(\|v(k)\|\,\|v^\star\|)$
\EndFor
\Statex \textbf{Phase 2: screen, then rank}
\State $\mathcal{K}'\gets\{\,k\in\mathcal{K}:\widehat{\mathrm{adv}}(k)\ge\tau\,\}$
  \Comment{filter}
\State rank $\mathcal{K}'$ in descending order of $\mathrm{Score}(k)$
\State \Return top-ranked candidates from $\mathcal{K}'$
\end{algorithmic}
\end{algorithm}

\subsubsection{Complexity Analysis}\label{sec:cap-complexity}

We analyse the data, time, and memory of the distinguisher and the key recovery.
Throughout, $n$ is the block size, $\kappa$ the last-round subkey length
($|\mathcal{K}|=2^\kappa$), $t$ the number of mask pairs used by the
distinguisher, $m$ the number of heavy tuples $(w_i,\alpha_i,\beta_i)$ returned by
\textsc{MACSample}, and $\tau$ the correlation threshold defining a heavy pair
($\cor^2\ge\tau$). All capacities and correlations are in the normalised
$[-1,1]$ scale of the convention above.

\paragraph{Distinguisher.}
\textsc{CapacityDistinguisher} (Algorithm~\ref{alg:cap-dist}) spends $O(N\cdot t)$
time --- $t$ inner products per sample --- and $O(2^t)$ memory for the histogram
over $\mathbb{F}_2^t$, which doubles with each added mask pair, so $t$ is kept
small. The binding constraint is statistical: distinguishing the cipher from the
random baseline with the $\chi^2$ capacity test requires
\begin{equation}\label{eq:data-dist}
  N \;=\; \Omega\!\left(\frac{\sqrt{2^t}}{\mathrm{Cap}(g^w)}\right),
  \qquad
  \mathrm{Cap}(g^w)=\sum_{S\neq0}\cor(B,w;\alpha_S,\beta_S)^2 ,
\end{equation}
the multidimensional-linear data complexity of Section~\ref{sec:background}; the
$\sqrt{2^t}$ is the variance penalty of a $2^t$-cell goodness-of-fit statistic
(the test statistic $N\cdot\widehat{\mathrm{Cap}}$ has noncentrality $N\,\mathrm{Cap}$
against a central $\chi^2_{2^t-1}$ of standard deviation $\Theta(\sqrt{2^t})$).

\begin{table}[h]
\centering
\begin{tabular}{ll}
\hline
\textbf{Resource} & \textbf{Complexity} \\
\hline
Data   & $N = \Omega\!\bigl(\sqrt{2^t}/\mathrm{Cap}(g^w)\bigr)$ chosen plaintext pairs \\
Time   & $O(N \cdot t)$ \\
Memory & $O(2^t)$ \\
\hline
\end{tabular}
\caption{Complexity of \textsc{CapacityDistinguisher}
  (Algorithm~\ref{alg:cap-dist}).}
\label{tab:dist-complexity}
\end{table}

\paragraph{Key recovery.}
\textsc{TemplateKeyRecovery} (Algorithm~\ref{alg:template-key-recovery}) estimates,
for each of the $2^\kappa$ candidate subkeys, the $m$ component correlations
$\cor_i(k)$ over the heavy tuples. Each estimate uses $N=O(1/\tau)$ pairs:
averaging $\pm1$ values, Hoeffding gives additive error $\sqrt\tau$ with constant
confidence, which resolves a heavy correlation ($|\cor_i^\star|\ge\sqrt\tau$) and,
crucially, its \emph{sign} for the template score. This distinct $N$ should not be
confused with the distinguisher's aggregate sample size~\eqref{eq:data-dist}. The
leading time term is therefore
\[
  O\!\left(2^\kappa \cdot m \cdot N\right) \;=\; O\!\left(2^\kappa\, m/\tau\right),
\]
one partial decryption $R_k^{-1}$ per (key, pair) and $m$ masked-bit reads each.
The ciphertext pairs are queried \emph{once} and reused across all candidates ---
only $R_k^{-1}$ is recomputed per key --- so the data cost is $O(W/\tau)$ chosen
plaintext pairs, where $W$ is the number of distinct differences among the tuples
($W=1$ in the common fixed-$w$ case, where a single batch of $O(1/\tau)$ pairs
serves all $m$ correlations). The screen and signed re-ranking cost
$O(2^\kappa\log 2^\kappa)$, dominated by the estimation term. Memory is
$O(2^\kappa+m)$: one score per candidate plus the working correlation vector.

\begin{table}[h]
\centering
\begin{tabular}{ll}
\hline
\textbf{Resource} & \textbf{Complexity} \\
\hline
Data   & $O(W/\tau)$ chosen plaintext pairs, $N=O(1/\tau)$ per correlation, reused across keys \\
Time   & $O\!\left(2^\kappa \cdot m / \tau\right)$ \\
Memory & $O\!\left(2^\kappa + m\right)$ \\
\hline
\end{tabular}
\caption{Complexity of \textsc{TemplateKeyRecovery}
  (Algorithm~\ref{alg:template-key-recovery}). $W\le m$ is the number of distinct
  input differences.}
\label{tab:kr-complexity}
\end{table}

\noindent
\emph{Offline, once.} Collecting the $m$ heavy tuples costs $O(m/\mathrm{adv})$
\textsc{MACSample} draws ($O(1)$ quantum queries each), and building the signed
reference $v^\star$ on the reduced cipher costs $O(m/\tau)$ classical queries (or
$O(m/\sqrt\tau)$ via amplitude estimation, Section~\ref{sec:mac-attacks}). Both
are paid once and amortized over all $2^\kappa$ candidates.

\paragraph{Tradeoffs.}
Each $R_k^{-1}$ is a single-round map, so the attack recovers the $\kappa$-bit
subkey with $O(2^\kappa\,m/\tau)$ one-round operations on data collected once ---
far below the cost of exhaustive search over the full key, and the standard payoff
of a last-round attack. The mask budget $m$ trades linearly: more heavy tuples
sharpen the separation between the correct key and the equivalence class but raise
both time and memory in proportion. The threshold $\tau$ trades in the opposite
direction: a larger $\tau$ admits only stronger approximations (fewer pairs per
correlation, $N=1/\tau$ smaller) but requires the cipher to exhibit heavier
correlations, while a smaller $\tau$ exploits weaker structure at higher data
cost. The bounds above are expected costs; a rigorous treatment would bound the
variance of the energy and template estimators and derive the $N$ needed for
target false-positive and false-negative rates, which we leave to future work.

\subsection{A MAC-Fishing-Based Block Cipher Distinguisher and Key Recovery}
\label{sec:mac-attacks}

The attacks of this subsection use the quantum sampler of
Section~\ref{sec:macfishing} and therefore operate in the \emph{Q2} model, in
which the attacker has superposition (quantum) oracle access to the
cipher~\cite{kaplan-q1q2,boneh-zhandry-q1q2}: \textsc{MACSample} applies $U_B$ to a
uniform superposition, so the mask-finding step has no classical analogue. Where
we report a ``classical verification'' cost we mean replacing the
amplitude-estimation substep with classical sampling against the same oracle ---
the fishing step still requires Q2 access. 

\subsubsection{MAC-Fishing-Based Block Cipher Distinguisher}\label{sec:mac-dist}

The distinguisher separates a structured cipher from a random permutation in two
phases. Phase~1 runs \textsc{MACSample} (Algorithm~\ref{algo:cd_w}), which with a
constant number of queries returns a candidate $(\alpha,\beta)$ drawn
$\propto\cor(B,w;\alpha,\beta)^2$; for a structured cipher it lands in the heavy
set $\{\cor^2\ge\tau\}$ with probability $\mathrm{adv}:=\mathrm{adv}(w,N\tau)$.
Fishing alone does not certify the candidate --- a single draw could be a spurious
low-correlation pair --- so Phase~2 \emph{verifies} it by estimating
$\cor(B,w;\alpha,\beta)$ and testing $\cor^2\ge\tau$. The full distinguisher is assembled in
Algorithm~\ref{alg:mac-fishing-dist}.

\begin{algorithm}[h]
\caption{\textsc{MACFishingDistinguisher}}
\label{alg:mac-fishing-dist}
\begin{algorithmic}[1]
\Require quantum oracle $\mathcal{O}_B$; difference $w$; threshold $\tau$;
  precision $\varepsilon=\sqrt{\tau}$
\Ensure ``cipher'' or ``random permutation''
\State $(\alpha,\beta) \gets$ \textsc{MACSample}$(\mathcal{O}_B, w)$
  \Comment{Algorithm~\ref{algo:cd_w}, $O(1)$ quantum queries}
\State $\widehat{\cor} \gets \textsc{QuantumVerify}(\mathcal{O}_B, w, \alpha,
  \beta, \varepsilon)$
  \Comment{or \textsc{ClassicalVerify} (App.~\ref{app:classical-mac})}
\If{$\widehat{\cor}^2 \ge \tau$}
  \State \Return ``cipher''
\Else
  \State \Return ``random permutation''
\EndIf
\end{algorithmic}
\end{algorithm}

Verification is where the precision speed-up enters. To place the decision at
$\tau$ on the squared correlation we need $\cor(B,w;\alpha,\beta)$ to additive
error $\Theta(\sqrt\tau)$. Classically this costs $O(1/\tau)$ samples (Hoeffding);
quantumly, amplitude estimation reaches the same precision in $O(1/\sqrt\tau)$
queries. Algorithm~\ref{alg:quantum-verify} states the quantum verifier: it
phase-encodes $f_{w,\alpha,\beta}(x)=\alpha\cdot B(x)\oplus\beta\cdot B(x\oplus w)$
into a uniform superposition, so the amplitude on $\ket{+^n}$ is exactly
$\tfrac1{2^n}\sum_x(-1)^{f_{w,\alpha,\beta}(x)}=\cor(B,w;\alpha,\beta)$, and reads
that signed amplitude off by a Hadamard test with amplitude estimation. When only
classical query access is available the verification falls back to classical
sampling (\textsc{ClassicalVerify}, Appendix~\ref{app:classical-mac}) at the
quadratically worse cost $O(1/\tau)$.

\begin{algorithm}[h]
\caption{\textsc{QuantumVerify}}
\label{alg:quantum-verify}
\begin{algorithmic}[1]
\Require quantum oracle $\mathcal{O}_B$ for $B$; difference $w$; mask pair
  $(\alpha,\beta)$; precision $\varepsilon$
\Ensure signed estimate $\widehat{\cor}(B,w;\alpha,\beta)$ to additive error $\varepsilon$
\State define the Boolean phase function
  $f_{w,\alpha,\beta}(x) \gets \alpha\cdot B(x)\oplus\beta\cdot B(x\oplus w)$
\State prepare $\ket{\psi} \gets \tfrac{1}{2^{n/2}}\sum_{x}
  (-1)^{f_{w,\alpha,\beta}(x)}\ket{x}$
  \Comment{$O(1)$ calls to $\mathcal{O}_B$; $\langle +^n|\psi\rangle=\cor(B,w;\alpha,\beta)$}
\State estimate the \emph{signed} overlap $\langle +^n|\psi\rangle$ to additive
  error $\varepsilon$ by a Hadamard test with amplitude estimation
  \Comment{$O(1/\varepsilon)$ queries}
\State \Return $\widehat{\cor}(B,w;\alpha,\beta)$
\end{algorithmic}
\end{algorithm}

Completeness of the distinguisher rests on the fishing advantage; soundness rests on a random
permutation having no heavy pair to find. For uniformly random $B$ and
$w\neq0^n$, $\cor(B,w;\alpha,\beta)\sim\N(0,2^{-n})$, so for any fixed nonzero pair
\[
  \Pr[\cor(B,w;\alpha,\beta)^2\ge\tau]
  = \Pr\!\big[|\N(0,1)|\ge\sqrt{\tau\,2^n}\big]
  \;\le\; e^{-\tau\,2^n/2}
\]
by the Gaussian tail bound. A union bound over the $2^{2n}$ mask pairs gives
\[
  \Pr\!\big[\exists(\alpha,\beta):\cor(B,w;\alpha,\beta)^2\ge\tau\big]
  \;\le\; 2^{2n}\,e^{-\tau\,2^n/2},
\]
negligible once $\tau=\omega(\log N/N)$ (e.g.\ $\tau\ge 5n\,2^{-n}$). The heavy set
is then empty for a random permutation with high probability, Phase~2 rejects
whatever Phase~1 returns, and the false-positive rate is negligible --- the
threshold regime announced in the convention above.

\paragraph{Complexity of the MAC-Fishing distinguisher.}
\begin{proposition}[Complexity of \textsc{MACFishingDistinguisher}]
\label{prop:mac-dist}
Let $\mathrm{adv}=\mathrm{adv}(w,N\tau)$ be the probability that one
\textsc{MACSample} draw satisfies $\cor^2\ge\tau$, and let $\tau=\omega(\log N/N)$.
Algorithm~\ref{alg:mac-fishing-dist} uses $O(1)$ quantum queries in Phase~1 and
$O(1/\sqrt\tau)$ in Phase~2, succeeding per run with probability $\mathrm{adv}$.
Amplifying by $O(1/\mathrm{adv})$ repetitions gives total query complexity
\begin{equation}\label{eq:mac-dist-total}
  O\!\left(\frac{1}{\sqrt\tau\,\mathrm{adv}}\right)\ \text{(quantum)},
  \qquad
  O\!\left(\frac{1}{\tau\,\mathrm{adv}}\right)\ \text{(classical verification)}.
\end{equation}
\end{proposition}

\begin{proof}
Phase~1 returns a candidate in $O(1)$ queries (Algorithm~\ref{algo:cd_w}), heavy
with probability $\mathrm{adv}$. Phase~2 estimates $\cor(B,w;\alpha,\beta)$ to error
$\sqrt\tau$: amplitude estimation in $O(1/\sqrt\tau)$ queries, classical sampling
in $O(1/(\sqrt\tau)^2)=O(1/\tau)$. The random-permutation false-positive bound is
the union bound above. Repeating $O(1/\mathrm{adv})$ times amplifies success to a
constant and yields~\eqref{eq:mac-dist-total}.
\end{proof}

The factor $1/\mathrm{adv}$ is the binding cost --- the draws needed to land a heavy
pair --- and is small exactly when the cipher concentrates its correlation mass.
Phase~1 has no classical analogue: locating a heavy pair is the MAC-Fishing
problem, classically hard by Theorem~\ref{thm:MAC_lb}. That theorem is proved in
the dense regime $N\tau=\Theta(1)$, whereas the distinguisher runs at the sparser
$N\tau=\omega(\log N)$; there the classical obstruction remains the exponential
search over the mask space, though the precise $\Omega(N/\log N)$ bound is stated
only for the dense threshold.

\noindent Table~\ref{tab:dist-complexity-mac} summarizes the complexities.

\begin{table}[h]
\centering
\renewcommand{\arraystretch}{1.4}
\begin{tabular}{lcc}
\hline
\textbf{Phase} & \textbf{Quantum queries} & \textbf{Classical verification} \\
\hline
Phase 1: \textsc{MACSample} (Alg.~\ref{algo:cd_w}) & $O(1)$ & --- \\
Phase 2: correlation verification & $O(1/\sqrt\tau)$ & $O(1/\tau)$ \\
\hline
Single run, total & $O(1/\sqrt\tau)$ & $O(1/\tau)$ \\
With repetition (succ.\ $\ge 2/3$)
  & $O\!\left(\tfrac{1}{\sqrt\tau\,\mathrm{adv}}\right)$
  & $O\!\left(\tfrac{1}{\tau\,\mathrm{adv}}\right)$ \\
\hline
\multicolumn{3}{l}{$\mathrm{adv}=\mathrm{adv}(w,N\tau)$;\ false positive negligible
  for $\tau=\omega(\log N/N)$;\ false negative $1-\mathrm{adv}$.} \\
\hline
\end{tabular}
\caption{Query complexity of \textsc{MACFishingDistinguisher}. }
\label{tab:dist-complexity-mac}
\end{table}

\subsubsection{Key Recovery on the $(r+1)$-Round Cipher}\label{sec:mac-kr}

We extend the distinguisher to a last-round key recovery on the $(r+1)$-round
cipher $E_K$, in two phases. The \emph{offline} phase, run once, builds a
correlation \emph{template}: $m$ heavy tuples $(w_i,\alpha_i,\beta_i)$ verified to
lie in $\{\cor^2\ge\tau\}$ for the $r$-round reduced cipher $B$, with their signed
reference correlations $\cor_i^\star=\cor(B,w_i;\alpha_i,\beta_i)$. The
\emph{online} phase, per candidate subkey $k$, forms the partially decrypted
oracle $B_k(x)=R_k^{-1}(E_K(x))$ and scores how well its profile matches the
template. When $k=k^\star$, $B_{k^\star}=B$ and the profile matches; wrong guesses
randomize (wrong-key randomization), giving correlations at the $\N(0,2^{-n})$
baseline.

As in the classical case (Section~\ref{sec:cap-kr}), a single heavy mask suffices
to \emph{distinguish}, but key recovery must separate $k^\star$ from the
equivalence class of wrong keys that share its magnitudes (hence its unsigned
energy) while differing in sign. We score with the signed cosine similarity
\begin{equation}\label{eq:cosine-mac}
  \mathrm{Score}(k)
  = \frac{\langle v(k),v^\star\rangle}{\|v(k)\|\,\|v^\star\|},
  \qquad
  \langle v(k),v^\star\rangle
    = \sum_{i=1}^m \widehat{\cor}_i(k)\,\cor_i^\star ,
\end{equation}
where $v(k)=(\widehat{\cor}_1(k),\dots,\widehat{\cor}_m(k))$ and
$\widehat{\cor}_i(k)$ is the signed estimate of $\cor(B_k,w_i;\alpha_i,\beta_i)$
from \textsc{QuantumVerify} (with classical query access only, the scoring falls
back to \textsc{ClassicalKeyScore}, Appendix~\ref{app:classical-mac}). The signed
reference $v^\star$ carries the inner-key-dependent signs and is obtained offline
from the reduced cipher, as discussed in Section~\ref{sec:cap-kr}.

\paragraph{Offline phase.}
A \textsc{MACSample} draw is heavy with probability
$\mathrm{adv}=\mathrm{adv}(w,N\tau)$, so collecting $m$ verified heavy tuples takes
$O(m/\mathrm{adv})$ iterations, each $O(1)$ to fish and $O(1/\sqrt\tau)$ to verify;
the offline cost is $O\!\big(m/(\mathrm{adv}\sqrt\tau)\big)$ quantum queries to
$\mathcal{O}_B$, paid once (Algorithm~\ref{alg:template-establishment}).

\begin{algorithm}[h]
\caption{\textsc{TemplateEstablishment}}
\label{alg:template-establishment}
\begin{algorithmic}[1]
\Require quantum oracle $\mathcal{O}_B$ for the $r$-round reduced cipher; threshold
  $\tau$; template size $m$; precision $\varepsilon=\sqrt\tau$
\Ensure template $\mathcal{T}=\{(w_i,\alpha_i,\beta_i)\}_{i=1}^m$ and reference
  $v^\star=(\cor_1^\star,\dots,\cor_m^\star)$
\State $\mathcal{T}\gets\emptyset$, \quad $v^\star\gets()$
\Repeat
  \State $(w,\alpha,\beta)\gets$ \textsc{MACSample}$(\mathcal{O}_B)$
    \Comment{joint sampler, Section~\ref{sec:macfishing}}
  \State $\widehat{\cor}\gets\textsc{QuantumVerify}(\mathcal{O}_B,w,\alpha,\beta,\varepsilon)$
  \If{$\widehat{\cor}^2\ge\tau$ \textbf{and} $(w,\alpha,\beta)\notin\mathcal{T}$}
    \State append $(w,\alpha,\beta)$ to $\mathcal{T}$ and $\widehat{\cor}$ to $v^\star$
  \EndIf
\Until{$|\mathcal{T}|=m$}
\State \Return $\mathcal{T}$, $v^\star$
\end{algorithmic}
\end{algorithm}

\paragraph{Online phase.}
Given $(\mathcal{T},v^\star)$, each candidate $k$ is scored by estimating $v(k)$ on
$B_k$ and applying~\eqref{eq:cosine-mac}. An unsigned energy pre-filter
$\widehat{\mathrm{adv}}(k)=\tfrac1m\sum_i\widehat{\cor}_i(k)^2\ge\tau$ discards the
non-structured guesses (template correlations at the $2^{-n}$ baseline), and the
survivors $\mathcal{K}'$ are re-ranked by the signed score.

\paragraph{The Attack} Algorithm~\ref{alg:mac-fishing-kr} assembles the offline and online
phases into complete attack
\begin{algorithm}[h]
\caption{\textsc{MACFishingKeyRecovery}}
\label{alg:mac-fishing-kr}
\begin{algorithmic}[1]
\Require $(r+1)$-round quantum oracle $\mathcal{O}_{E_K}$; $r$-round oracle
  $\mathcal{O}_B$ (offline); candidate subkeys $\mathcal{K}$, $|\mathcal{K}|=2^\kappa$;
  threshold $\tau$; template size $m$; precision $\varepsilon=\sqrt\tau$
\Ensure ranked list of last-round subkey candidates
\State $(\mathcal{T},v^\star)\gets\textsc{TemplateEstablishment}(\mathcal{O}_B,\tau,m,\varepsilon)$
  \Comment{offline}
\For{each $k\in\mathcal{K}$}
  \State $v(k)\gets ()$
  \For{$i = 1$ \textbf{to} $m$}
    \State define $\mathcal{O}_{B_k}$ for $B_k(x)=R_k^{-1}(E_K(x))$
    \State $\widehat{\cor}_i(k)\gets\textsc{QuantumVerify}(\mathcal{O}_{B_k},w_i,\alpha_i,\beta_i,\varepsilon)$;
      append to $v(k)$
      \Comment{classical: \textsc{ClassicalKeyScore} (App.~\ref{app:classical-mac})}
  \EndFor
  \State $\widehat{\mathrm{adv}}(k)\gets\tfrac1m\sum_i\widehat{\cor}_i(k)^2$;\quad
    $\mathrm{Score}(k)\gets\langle v(k),v^\star\rangle/(\|v(k)\|\,\|v^\star\|)$
\EndFor
\State $\mathcal{K}'\gets\{\,k\in\mathcal{K}:\widehat{\mathrm{adv}}(k)\ge\tau\,\}$
  \Comment{energy pre-filter}
\State rank $\mathcal{K}'$ in descending order of $\mathrm{Score}(k)$
\State \Return top-ranked candidates from $\mathcal{K}'$
\end{algorithmic}
\end{algorithm}

\paragraph{Complexity of the MAC-Fishing Key Recovery.}
\begin{proposition}[Complexity of \textsc{MACFishingKeyRecovery}]
\label{prop:mac-kr}
With $\varepsilon=\sqrt\tau$ and $\mathrm{adv}=\mathrm{adv}(w,N\tau)$, the offline
phase costs $O\!\big(m/(\mathrm{adv}\sqrt\tau)\big)$ quantum queries to
$\mathcal{O}_B$ and the online phase $O(2^\kappa m/\sqrt\tau)$ quantum queries to
$\mathcal{O}_{E_K}$, for a total of
\begin{equation}\label{eq:mac-kr-total}
  O\!\left(\Big(2^\kappa+\tfrac{1}{\mathrm{adv}}\Big)\frac{m}{\sqrt\tau}\right)
  \ \text{(quantum)},
  \qquad
  O\!\left(\Big(2^\kappa+\tfrac{1}{\mathrm{adv}}\Big)\frac{m}{\tau}\right)
  \ \text{(classical verification)}.
\end{equation}
When $1/\mathrm{adv}\le 2^\kappa$ the offline term is dominated and the total is
$O(2^\kappa m/\sqrt\tau)$; the quantum improvement over classical verification is
the quadratic $1/\sqrt\tau$ factor, with the $2^\kappa$ enumeration unchanged
(Remark~\ref{grrem}).
\end{proposition}

\begin{proof}
Template establishment performs $O(m/\mathrm{adv})$ fishing iterations to collect
$m$ verified heavy tuples (each draw heavy with probability $\mathrm{adv}$), at
$O(1)$ to fish and $O(1/\sqrt\tau)$ to verify per iteration, giving
$O(m/(\mathrm{adv}\sqrt\tau))$. Online, \textsc{QuantumVerify} estimates the $m$
fixed template correlations per key at $O(1/\sqrt\tau)$ each over $2^\kappa$ keys
(no fishing), giving $O(2^\kappa m/\sqrt\tau)$; classically each estimate is
$O(1/\tau)$. Summing yields~\eqref{eq:mac-kr-total}.
\end{proof}

\begin{table}[h]
\centering
\renewcommand{\arraystretch}{1.4}
\begin{tabular}{lcc}
\hline
\textbf{Phase} & \textbf{Quantum queries} & \textbf{Classical verification} \\
\hline
Offline: template establishment & $O(m/(\mathrm{adv}\sqrt\tau))$ & $O(m/(\mathrm{adv}\tau))$ \\
Online: score each $k\in\mathcal{K}$ & $O(2^\kappa m/\sqrt\tau)$ & $O(2^\kappa m/\tau)$ \\
Online: ranking & $O(|\mathcal{K}'|)$ & $O(|\mathcal{K}'|)$ \\
\hline
\textbf{Total} & $O\!\big((2^\kappa+\tfrac1{\mathrm{adv}})\,m/\sqrt\tau\big)$
  & $O\!\big((2^\kappa+\tfrac1{\mathrm{adv}})\,m/\tau\big)$ \\
\hline
\end{tabular}
\caption{Query complexity of \textsc{MACFishingKeyRecovery}.}
\label{tab:mac-kr-complexity}
\end{table}

\begin{remark}[Grover-accelerated key search]\label{grrem}
In the Q2 model the linear scan over $\mathcal{K}$ can be replaced by Grover search
with a marking oracle testing the energy pre-filter
$\widehat{\mathrm{adv}}(k)\ge\tau$, reducing the enumeration from $2^\kappa$ to
$\sqrt{2^\kappa}$ for a combined $\sqrt{2^\kappa/\tau}$ speed-up. Two caveats.
First, the energy test is sign-blind, so Grover returns a member of the
high-energy equivalence class rather than $k^\star$; the signed
disambiguation~\eqref{eq:cosine-mac} must still be applied among the marked
candidates. Second, the marking oracle nests amplitude estimation --- a
bounded-error, measurement-based subroutine --- inside Grover, so it must be made
coherent and amplified to failure probability $o(1/\sqrt{2^\kappa})$ (e.g.\ by
probability filtering~\cite{bera_low-space_2024} or a QSVT
reformulation~\cite{Gilyen2019}) for the diffusion to remain valid. We leave a
rigorous treatment of this composition to future work and state the
amplitude-estimation-only bound as our result.
\end{remark}

%% file: final/appendix_lower-bound-proof.tex
\section{Proof of Theorem~\ref{thm:MAC_lb}}\label{sec:proof-MACfishing-lb}

\paragraph{Preface:}
The proof of the theorem closely follows the steps of Theorem 9 of Aaronson and Chen~\cite{aaronson-chen}. That theorem shows a similar hardness for the promise Fourier Fishing problem. In that problem the objective is to ``fish'' some $z$ with a high-valued Fourier coefficient $|\hat{f}(z)|^2 \ge 1$; in contrast, the objective of MAC Fishing is to output some $(\alpha, \beta)$ such that $|C_w(\alpha, \beta)|^2 \ge \tau$.

Below we present the proof for the sake of completeness; in particular, we point out the differences that arise due to masked cross-correlation. We follow a similar line of reasoning and many steps and results will seem repetitive -- those are included in entirety to have the entire proof at one place and further, allow a reader to apply the approach on other similar sampling problems.

\begin{proof}[Proof of Theorem]
The entire proof relies on Lemma~\ref{lemma:lem30} and Lemma~\ref{lemma:lem31} which are given below.

By way of contradiction, suppose there is a randomized algorithm $A$ that solves the promise version of MAC Fishing with $Succ_R + \Omega(1)$ probability yet makes only $o(N/\log N)$ queries to $B$.

Lemma~\ref{lemma:lem30} says that when $B$ is randomly chosen from all $n$-bit permutations, $\adv(w,\tau)$ is at least $Succ_Q - \tfrac{1}{n}$ with at least $1-\tfrac{48.5n^2}{N^2}$ probability. Thus, $1-o(1)$ fraction of $n$-bit permutations satisfy $\adv(w,\tau) \ge Succ_Q - \tfrac{1}{n}$. These fractions satisfy the promise stated in the theorem, and thus, $A$ can solve all of them using $o(N/\log N)$ queries and with a success probability $Succ_Q + \Omega(1)$.

Thus, when $B$ is randomly chosen from all $n$-bit permutations, as is done in Lemma~\ref{lemma:lem31}, $A$ can solve the MAC Fishing problem on it using $o(N/\log N)$ queries with probability at least
$$(1-o(1)) \cdot (Succ_R + \Omega(1)) = Succ_R + \Omega(1).$$

This contradicts Lemma~\ref{lemma:lem31}.
\end{proof}

Next we present the key technical lemmas required for the above theorem. We require the following facts and notations to simplify certain expressions. $\tau$ should be considered as  positive in this section.

\begin{itemize}
    \item $\displaystyle \frac{1}{\sqrt{2\pi}}\int_{-\infty}^{t} 
    x^2 e^{-x^2/2}\,dx = \Phi(t) - t \phi(t)$, where $\Phi(t)$ is the standard normal CDF and $\phi(t)$ is the standard normal PDF
    \item (Total second moment) $\displaystyle \hat{\Phi}(t) := \tfrac{1}{\sqrt{2\pi}} \int_{-\infty}^{t} x^2 e^{-x^2/2} dx = \tfrac{1}{\sqrt{2\pi}} \int_{-t}^{+\infty} x^2 e^{-x^2/2} dx = \Phi(t) - \tfrac{t}{\sqrt{2\pi}} e^{-t^2/2}$
    \item $\displaystyle \mathrm{Succ}_Q(\tau) 
    = \frac{2}{\sqrt{2\pi}}\int_{\sqrt{\tau}}^{\infty} 
    x^2 e^{-x^2/2}\,dx
    = 2\left(1 - \hat{\Phi}(\sqrt{\tau})\right)$ which takes values in $(0,2)$. We will use the fact that $\mathrm{Succ}_Q(\tau)$ is decreasing in $\tau$.
    \item $\displaystyle \Delta(\tau) 
    = 4\left(1-\hat{\Phi}\!\left(\sqrt{\tau/2}\right)\right)
    - 2\left(1-\hat{\Phi}(\sqrt{\tau})\right)
    = 2+ 2\hat{\Phi}(\sqrt{\tau}) 
    - 4\hat{\Phi}\!\left(\sqrt{\tau/2}\right)$
    \begin{itemize}
        \item $\Delta(\tau) > 0$. This is since $\Delta(\tau)$
            \begin{align*}
    & = 4\left(1-\hat{\Phi}(\sqrt{\tau/2})\right)
    - 2\left(1-\hat{\Phi}(\sqrt{\tau})\right) \\
    & = 2+2\hat{\Phi}(\sqrt{\tau}) 
    - 4\hat{\Phi}\!\left(\sqrt{\tau/2}\right) \\
    & = 2\left(1 -  \hat{\Phi}(\sqrt{\tau/2})\right) + 2 \left( \hat{\Phi}(\sqrt{\tau}) - \hat{\Phi}(\sqrt{\tau/2})\right) \\
    & > 0 \tag{$\because$ $\hat{\Phi}$ is monotonically increasing and $\sqrt{\tau/2} < \sqrt{\tau}$}
    \label{eq:delta_tau}
\end{align*}
        \item $\Delta(\tau) < 2$. This is since (i) $\hat{\Phi}(t) < 1$, and so $2 \hat{\Phi}(\sqrt{\tau}) < 2$, and (ii) since the second moment involves an even function, for positive $\tau$, $\hat{\Phi}(\sqrt{\tau/2}) \ge 1/2$, and so $4\hat{\Phi}(\sqrt{\tau/2}) \ge 2$.
    \end{itemize}
\end{itemize}



\begin{lemma}\label{lemma:lem30}
    The following holds for large $n$, positive $\tau$, and a randomly chosen permutation $B$.
    \begin{enumerate}
        \item $\displaystyle \E_B[\adv(w,\tau)] =    \mathrm{Succ}_Q(\tau) + \tfrac{\Delta(\tau)}{N} \text{ which is between } \mathrm{Succ}_Q(\tau)$ and $\mathrm{Succ}_Q(\tau) + 2/N$.
        \item $\displaystyle \Pr_{B \leftarrow \mathcal{P}_n} \left[ \adv(w,\tau) \le \mathrm{Succ}_Q(\tau) - \tfrac{1}{n} \text{ OR } \adv(w,\tau) \ge \mathrm{Succ}_Q(\tau) + \tfrac{2}{n} \right] \le \tfrac{48.5n^2}{N}$.
    \end{enumerate}
\end{lemma}


\begin{proof}
    Let's first compute $\E[\adv(w,\tau)]$.
    
    \paragraph{Distribution of $\adv$:} Recall that $C_w(\alpha,\beta)$ behaves like $\N(0,1)$ if $\alpha \not= \beta$ and like $\N(0,2)$ if $\alpha = \beta$ (Lemma~\ref{lem:cd_w_normal_dist_case2}).
    
    Next define 
    $$\delta(\alpha,\beta) =
    \begin{cases}
    1 & \text{ if } (\alpha,\beta) \in V^\tau_w\\
    0 & \text{ otherwise 0}.
    \end{cases}$$
    
    This allows us to write
    \begin{align*}
        \adv(w,\tau) &= \E_{\alpha,\beta} [C_w(\alpha,\beta)^2 \delta(\alpha,\beta)], ~
    \mathrm{and} \\
    \E_B[\adv(w,\tau)] &= \E_B \E_{\alpha,\beta} [C_w(\alpha,\beta)^2 \delta(\alpha,\beta)]\\
    & = \E_{\alpha,\beta} \E_B [C_w(\alpha,\beta)^2 \delta(\alpha,\beta)].
    \end{align*}
    
    Since $C_w(.,.)$ behaves like a normal variable and $\delta(\alpha, \beta)=1$ iff that normally distributed variable is outside of $[-\sqrt{\tau},\sqrt{\tau}]$, $C_w(\alpha,\beta)\delta(\alpha,\beta)$ is random variable that behaves like $\N(0,d)$ outside of the range and 0 inside; here, $d=1$ if $\alpha \not= \beta$ and 2 if $\alpha=\beta$.

    \paragraph{Expectation of $\adv$:} Consider the case of $\alpha \not= \beta$ which happens with probability $1- \tfrac{1}{N}$.
    \begin{align*}
        \mathbb{E}_B\!\left[C_w(\alpha,\beta)^2\,\delta(\alpha,\beta)
    \right]\bigg|_{\alpha \neq \beta}
    &= \frac{1}{\sqrt{2\pi}}\int_{|x| \geq \sqrt{\tau}} 
    x^2\,e^{-x^2/2}\,dx 
    = \frac{2}{\sqrt{2\pi}}
    \int_{\sqrt{\tau}}^{\infty} x^2\,e^{-x^2/2}\,dx = Succ_Q(\tau).
    \end{align*}

    Next, consider the case of $\alpha = \beta$ which happens with probability $\tfrac{1}{N}$. Here, $C_w(\alpha,\beta) \sim \mathcal{N}(0,2)$ with probability density $\frac{1}{\sqrt{4\pi}}e^{-x^2/4}$.
    

    \begin{align*}
    \mathbb{E}_B\!\left[
    C_w(\alpha,\beta)^2\,\delta(\alpha,\beta)
    \right]\bigg|_{\alpha = \beta}
    &= \frac{2}{\sqrt{4\pi}}
    \int_{\sqrt{\tau}}^{\infty} 
    x^2\,e^{-x^2/4}\,dx \notag\\
    &= \frac{2}{\sqrt{4\pi}}
    \int_{\sqrt{\tau/2}}^{\infty}
    2s^2\,e^{-s^2/2}\cdot\sqrt{2}\,ds \notag\\
    &= \frac{4\sqrt{2}}{\sqrt{4\pi}}\int_{\sqrt{\tau/2}}^{\infty}
    s^2\,e^{-s^2/2}\,ds\\
    &= 4\left(1 - \hat{\Phi}\!\left(
    \sqrt{\tau/2}\right)\right).
    \label{eq:exp_eq}
\end{align*}

    Thus, 
    \begin{align}
    \mathbb{E}_B[\mathrm{adv}(w,\tau)]
    &= \frac{N-1}{N}\cdot
    \mathrm{Succ}_Q(\tau)
    + \frac{1}{N}\cdot 
    4\left(1-\hat{\Phi}\!\left(\sqrt{\tau/2}\right)\right)
    \notag\\
    &= \mathrm{Succ}_Q(\tau) 
    + \frac{\Delta(\tau)}{N},
    \end{align}

Recall that,
    \begin{align*}
    \Delta(\tau) 
    & = 4\left(1-\hat{\Phi}(\sqrt{\tau/2})\right)
    - 2\left(1-\hat{\Phi}(\sqrt{\tau})\right) \\
    & = 2+2\hat{\Phi}(\sqrt{\tau}) 
    - 4\hat{\Phi}\!\left(\sqrt{\tau/2}\right) \\
    & = 2\left(1 -  \hat{\Phi}(\sqrt{\tau/2})\right) + 2 \left( \hat{\Phi}(\sqrt{\tau}) - \hat{\Phi}(\sqrt{\tau/2})\right) \\
    & > 0 \tag{$\because$ $\hat{\Phi}$ is monotonically increasing and $\sqrt{\tau/2} < \sqrt{\tau}$}
    \label{eq:delta_tau}
\end{align*}

Therefore, $\mathbb{E}[\mathrm{adv}(w,\tau)] 
> \mathrm{Succ}_Q(\tau)$.

    \paragraph{Variance of $\adv$:} Next, we will compute $\Var[\adv(w,\tau)] = \E[\adv(w,\tau)^2] - \E[\adv(w,\tau)]^2$. We already have the subtrahend:
    $$\E[\adv(w,\tau)]^2 \ge  Succ_Q^2(\tau)$$

    So, let's compute the minuend $\E[\adv(w,\tau)^2]$.

    \begin{align*}
        \E_B[\adv(w,\tau)^2] & = \E_B \large( \E_{\alpha,\beta} [C_w(\alpha,\beta)^2 \delta(\alpha,\beta)] \large)^2 \\
        & = \E_B \Big( \E_{\alpha,\beta,\alpha',\beta'} [C_w(\alpha,\beta)^2 \delta(\alpha,\beta) C_w(\alpha',\beta')^2 \delta(\alpha',\beta')] \Big)\\
        & = \E_{\alpha,\beta,\alpha',\beta'} \E_B [C_w(\alpha,\beta)^2 \delta(\alpha,\beta) C_w(\alpha',\beta')^2 \delta(\alpha',\beta')]
    \end{align*}

    Let's focus on $\E_B [C_w(\alpha,\beta)^2 \delta(\alpha,\beta) C_w(\alpha',\beta')^2 \delta(\alpha',\beta')]$.

There are multiple possibilities here.

\paragraph{Case $(\alpha, \beta) = (\alpha', \beta')$:} In this case, which happens with probability $1/N^2$, $C_w(\alpha',\beta') = C_w(\alpha, \beta)$ and $\delta(\alpha,\beta) = \delta(\alpha', \beta')$. 
If $\alpha=\beta$, $C_w(\alpha,\beta) \sim \N(0,2)$, otherwise, $C_w(\alpha, \beta) \sim \N(0,1)$. Using the fact that $3\sigma^4$ is the fourth moment of $\N(0,\sigma^2)$, we can bound the desired expectation for different sub-cases.

Thus,
\begin{align*}
    \E_B [C_w(\alpha,\beta)^2 \delta(\alpha,\beta)C_w(\alpha',\beta')^2 \delta(\alpha',\beta')] & = \E_B [C_w(\alpha,\beta)^4 \delta(\alpha,\beta)] \le 12.
\end{align*}


\textit{Sub-case}~$\alpha \neq \beta:$
\begin{align*}
    \E_B [C_w(\alpha,\beta)^2 \delta(\alpha,\beta)C_w(\alpha',\beta')^2 \delta(\alpha',\beta')]
    &= \frac{2}{\sqrt{2\pi}}
    \int_{\sqrt{\tau}}^{\infty} x^4\,e^{-x^2/2}\,dx\\
    &= \frac{2\tau^{3/2} \,e^{-\tau/2}}{\sqrt{2\pi}} 
    + 6\left(1 - \hat{\Phi}(\sqrt{\tau})\right)\\
    &= \frac{2\tau^{3/2}\,e^{-\tau/2}}{\sqrt{2\pi}} 
    + 3\cdot Succ_Q(\tau)\\
    &=K_1(\tau) \tag{Definition of $K_1(\tau)$}
\end{align*}
where, the integration by parts is done using $u=x^3, v=-e^{x^2/2}$.
Using that fact that $\mathrm{Succ}_Q(\tau) \le 2$ and $\tau^{3/2} e^{-\tau/2}$ attains a maximum value less than 1.2 for positive $\tau$, it can be verified that $K_1(\tau) \le 7$.

\textit{Sub-case}~$\alpha = \beta:$
\begin{align*}
    \E_B [C_w(\alpha,\beta)^2 \delta(\alpha,\beta)C_w(\alpha',\beta')^2 \delta(\alpha',\beta')]
    &= \frac{2}{\sqrt{4\pi}}
    \int_{\sqrt{\tau}}^{\infty} x^4\,e^{-x^2/4}\,dx 
    \label{eq:K2_def}\\
    &= \frac{4\sqrt{2}}{\sqrt{\pi}}
    \left[\left(\frac{\tau}{2}\right)^{3/2} 
    \right]e^{-\tau/4} 
    + 12
    \left(1 - \hat{\Phi}\!\left(\sqrt{\tau/2}\right)\right)\\
    &=\frac{4\sqrt{2}}{\sqrt{\pi}}
    \left[\left(\frac{\tau}{2}\right)^{3/2} 
    \right]e^{-\tau/4} + 6~\mathrm{Succ}_Q(\tau/2)\\
    &= K_2(\tau) \tag{Definition of $K_2(\tau)$}
\end{align*}
Using that fact that $\mathrm{Succ}_Q(\tau) \le 2$ and $\tau^{3/2} e^{-\tau/4}$ attains a maximum value less than 3.7 for positive $\tau$, it can be verified that $K_2(\tau) \le 16.2$.

Therefore, combining both sub-cases for case $(\alpha, \beta)=(\alpha', \beta')$, we have: 
\begin{align*}
    \mathbb{E}_B\!\left[
    C_w(\alpha,\beta)^4\,\delta(\alpha,\beta)
    \right]
    = K_1(\tau) 
    + \frac{K_2(\tau)-K_1(\tau)}{N}.
\end{align*}

\paragraph{Case $(\alpha, \beta) \not= (\alpha', \beta')$:} This case happens with probability $(1-\tfrac{1}{N^2})$.

Define two sets A and D as follows:
\begin{align*}
A &= \{x : \alpha \cdot B(x) = \alpha' \cdot B(x),~ \beta \cdot B(x \xor w) = \beta' \cdot B(x \xor w)\} \\
&\quad \cup \{x: \alpha \cdot B(x) \neq \alpha' \cdot B(x),\; \beta \cdot B(x \xor w) \neq \beta' \cdot B(x \xor w)\}, \\[6pt]
D &= \{x : \alpha \cdot B(x) = \alpha' \cdot B(x),\; \beta \cdot B(x \xor w) \neq \beta' \cdot B(x \xor w)\} \\
&\quad \cup \{x: \alpha \cdot B(x) \neq \alpha' \cdot B(x),\; \beta \cdot B(x \xor w) = \beta' \cdot B(x \xor w)\}, \\
\end{align*}



Observe that $A$ and $D$ are disjoint and $|A| +|D|=N$ with $|A| = |D|= N/2$ for a random permutation.  Use these sets to define the following sums.
\begin{align*}
C_A{(\alpha, \beta)} &= \frac{1}{\sqrt{N}} \sum_{x \in A} (-1)^{\alpha \cdot B(x)} \cdot (-1)^{\beta \cdot B(x \xor w)}\\
&=\frac{1}{\sqrt{N}} \sum_{x \in A} (-1)^{\alpha' \cdot B(x)} \cdot (-1)^{\beta' \cdot B(x \xor w)}\\
C_D{(\alpha, \beta)} &= \frac{1}{\sqrt{N}}\sum_{x \in D} (-1)^{\alpha \cdot B(x)} \cdot (-1)^{\beta \cdot B(x \xor w)}\\
&=-\frac{1}{\sqrt{N}} \sum_{x \in D} (-1)^{\alpha' \cdot B(x)} \cdot (-1)^{\beta' \cdot B(x \xor w)}
\end{align*}

\paragraph{Distribution of $C_A$ and $C_D$}
Let's understand these random variables. First, observe that $C_A{(\alpha, \beta)}$ and $C_D{(\alpha, \beta)}$ are independent since they depend on disjoint inputs ($A$ vs $D$). Next, for a random $B$, both following normal distribution with 0 mean. They both have the same variance which is $1$ if $\alpha \not= \beta$ and $2$ if $\alpha = \beta$. Let $\sigma^2$ denote this variance.

Now, let's write
\[
    C_w{(\alpha,\beta)} = C_A{(\alpha, \beta)} + C_D{(\alpha, \beta)}, \quad C_w{(\alpha',\beta')} =  C_A{(\alpha, \beta)} -C_D{(\alpha, \beta)}.
\]

Observe that $|C_w(\alpha, \beta)| \ge \sqrt{\tau}$ implies that $|C_A + C_D| \ge \sqrt{\tau}$; we get a similar bound for the other expression as well.



Next, let's see how to compute the expectation of the following multivariate expression.
$$C_w(\alpha,\beta)^2 \delta(\alpha,\beta)C_w(\alpha',\beta')^2 \delta(\alpha',\beta') =
\begin{cases}
    (C_A + C_D )^2 (C_A - C_D)^2 & \text{ if $|C_A \pm C_D| \ge  \sqrt{\tau}$}\\
    0 & \text{ otherwise }
\end{cases}
$$

Denoting $C_A$ by $A$ and $C_D$ by $D$, 
the integral we want to evaluate is
\[
  \int_{|A+D|\geq \sqrt\tau \atop |A-D|\geq \sqrt\tau}\,(A + D)^2 (A-D)^2 f(A-D) f(A+D) \,dAdD
\]
where $f(u)$ is the probability density of $u\sim N(0,2\sigma^2).$ Two subcases arise due to two different $\sigma^2$ in lemma~\ref{lem:cd_w_normal_dist_case2}. So, we will have two different density functions.

\textbf{Subcase 1:} $\alpha\neq\beta$

In this case, $A, D\sim  \N(0,1/2)$ which implies that $A\pm D \sim \N(0,1)$ with individual marginal probability density function $f(A\pm D)=\frac{1}{\sqrt{2\pi}\cdot 1}e^{-(A\pm D)^2/2}$.

   \begin{align*}
       \E_B [C_w(\alpha,\beta)^2 \delta(\alpha,\beta)C_w(\alpha',\beta')^2 \delta(\alpha',\beta')] \Big|_{\alpha\neq\beta} &= \frac{1}{2\pi} \int_{ |A+D| \ge {\sqrt{\tau}} \atop |A-D| \ge {\sqrt \tau}}(A+D)^2 (A-D)^2 e^{-(A^2+D^2)/2} \cdot dA dD.
   \end{align*}
Taking $x= A+D$, $y=A-D$, we get $x^2+y^2=2(A^2+D^2).$ Thus, the above equation can be written as
\begin{align*}
    & \frac{1}{2\pi} \cdot \int_{ |A+D| \ge {\sqrt{\tau}} \atop |A-D| \ge {\sqrt{\tau}}} (A+D)^2 \cdot (A-D)^2 \cdot e^{-(A^2+D^2)/2} \cdot dA dB \\
    =& \int_{|x| \ge \sqrt{\tau} \atop |y| \ge \sqrt{\tau}} \frac{1}{2\pi}x^2 y^2 e^{-(x^2+y^2)/4} \cdot \frac{1}{2} dx dy\\
    =& \frac{1}{4\pi}\left(\int_{|x| \ge \sqrt{\tau}}x^2 e^{-x^2/4} dx \right)^2\\
    =& \frac{1}{4\pi}\left( \int_{\sqrt{\tau}}^{+\infty} 2x^2 e^{-x^2/4} dx \right)^2\\
    =&\frac{1}{4\pi}\left( \int_{\sqrt{\tau/2}}^{+\infty} 4t^2 e^{-t^2/2} \sqrt{2}dt \right)^2~~(\mathrm{take}~t=x/\sqrt{2})\\
    =& 4~\mathrm{Succ}_Q^2(\tau/2)
\end{align*}
Here, we have used the fact that $\mathrm{Succ}_Q(\tau) 
    = \frac{2}{\sqrt{2\pi}}\int_{\sqrt{\tau}}^{\infty} 
    x^2 e^{-x^2/2}\,dx\implies ~\int_{\sqrt{\tau/2}}^{\infty} 
    x^2 e^{-x^2/2}\,dx=\sqrt{\frac{\pi}{2}}\mathrm{Succ}_Q(\tau/2). $

\textbf{Subcase 2:} $\alpha=\beta$ \\
In this case, $A, D\sim  N(0,1)$ which implies that $A \pm D\sim N(0,2)$ with individual marginal probability density function $f(A\pm D)=\frac{1}{\sqrt{4\pi}} e^{-(A\pm D)^2/4}$. Then, the expectation becomes:

\begin{align*}
\E_B [C_w(\alpha,\beta)^2 \delta(\alpha,\beta)C_w(\alpha',\beta')^2 \delta(\alpha',\beta')] \Big|_{\alpha=\beta}
    &=\frac{1}{4\pi} \cdot \int_{ |A+B| \ge {\sqrt{\tau}} \atop |A-B| \ge {\sqrt{\tau}}} (A+D)^2 \cdot (A-D)^2 \cdot e^{-(A^2+B^2)/4} \cdot dA dD \\
    &= \frac{1}{4\pi} \int_{|x| \ge \sqrt{\tau} \atop |y| \ge \sqrt{\tau}} x^2 y^2 e^{-(x^2+y^2)/8} \cdot \frac{1}{2} dx dy\\
    &= \frac{1}{8\pi} \left(\int_{|x| \ge \sqrt{\tau}} x^2 e^{-x^2/8} dx \right)^2\\
    &= \frac{1}{8\pi} \left(\int_{\sqrt{\tau}}^\infty 2x^2 e^{-x^2/8} dx \right)^2\\
    &= \frac{1}{2\pi} \left(\int_{\sqrt{\tau}}^\infty x^2 e^{-x^2/8} dx \right)^2\\
    &= \frac{1}{2\pi}\left(\int_{\sqrt{\tau/4}}^\infty 4t^2 e^{-4t^2/8} \cdot2~dt \right)^2~(\mathrm{take}~t=x/2)\\
    &=\frac{1}{2\pi}\left( 4\sqrt{2\pi}\mathrm{Succ}_Q(\tau/4)\right)^2\\
    &= 16 \cdot \mathrm{Succ}_Q^2(\tau/4)
\end{align*}


Now, combining the two cases, we have:
\begin{align*}
	{\E_B[\adv(w,\tau)^2] }&= {\frac{1}{N^2} \Big[ {K_1(\tau)}+ \frac{K_2(\tau)-K_1(\tau)}{N}\Big] +\Big( 1-\frac{1}{N^2}\Big) \Big[ \frac{1}{N}\cdot 16~ \mathrm{Succ}_Q^2(\tau/4)+\left(1-\frac{1}{N}\right)\cdot4~\mathrm{Succ}^2_Q(\tau/2)\Big]}\\
 &\le \frac{1}{N^2}\left[K_1(\tau)+\frac{K_2(\tau)-K_1(\tau)}{N} \right]+ \left( 1-\frac{1}{N^2} \right)\left[ \frac{1}{N}\cdot16~ \mathrm{Succ}_Q^2(\tau/4) + 4\left(1-\frac{1}{N}\right)\mathrm{Succ}_Q^2(\tau/4)\right]\\
    &= \frac{1}{N^2}\left[K_1(\tau)+\frac{K_2(\tau)-K_1(\tau)}{N} \right]+ \left( 1-\frac{1}{N^2} \right)\left[ \frac{12}{N}\mathrm{Succ}^2_Q(\tau/4) + 4 \mathrm{Succ}^2_Q(\tau/4)\right]\\ 
      &= \frac{1}{N^2}\left[K_1(\tau)+\frac{K_2(\tau)-K_1(\tau)}{N} \right]+ 4\left( 1-\frac{1}{N^2} \right)\left(1+\frac{3}{N}\right)\mathrm{Succ}^2_Q(\tau/4)\\
\end{align*}
Since, by definition of $\mathrm{Succ}^2_Q(\tau)$ we have  $\mathrm{Succ}^2_Q(\tau)\le \mathrm{Succ}^2_Q(\tau/2)\le \mathrm{Succ}^2_Q(\tau/4)$, we get the following bound on variance. 
\begin{align*}
    \Delta(\tau)&= 4\left(1-\hat{\Phi}\!\left(\sqrt{\tau/2}\right)\right)
    - 2\left(1-\hat{\Phi}(\sqrt{\tau})\right)\\
    &=2~\mathrm{Succ}_Q(\tau/2)-\mathrm{Succ}_Q(\tau)\\
    & \ge2~\mathrm{Succ}_Q(\tau)- \mathrm{Succ}_Q(\tau) = \mathrm{Succ}_Q(\tau)
\end{align*}
From \eqref{eq:exp_eq}, we have $\E_B[\adv(w,\tau)] = \mathrm{Succ}_Q(\tau) + \frac{\Delta(\tau)}{N}\ge 2~\mathrm{Succ}_Q(\tau) \ge 2~ \mathrm{Succ}_Q(\tau/4)$.    
    \begin{align*}
    \Var[\adv(w,\tau)] &= \E[\adv(w,\tau)^2] - \E[\adv(w,\tau)]^2\\ &\leq \frac{1}{N^2}\left[K_1(\tau)+\frac{K_2(\tau)-K_1(\tau)}{N} \right]+ 4\left( 1-\frac{1}{N^2} \right)\left(1+\frac{3}{N}\right)\mathrm{Succ}^2_Q(\tau/4) - \Big( 2~\mathrm{Succ}_Q(\tau/4)\Big)^2\\
    & \leq \frac{1}{N^2}\left[K_1(\tau)+\frac{K_2(\tau)-K_1(\tau)}{N} \right]+\frac{12}{N}~\mathrm{Succ}^2_Q(\tau/4)\\
    \end{align*}

\paragraph{Concentration bound:} 
We have already proved that (Equation~\ref{eq:exp_eq}) 
$$\E_B[\adv(w,\tau)] \in (\mathrm{Succ}_Q(\tau), \mathrm{Succ}_Q(\tau) + 2/N).$$ We want to use Chebyshev's inequality to bound the probability of $\adv(w,\tau)$ lying outside a larger range, namely $\adv \le \mathrm{Succ}_Q(\tau) - \tfrac{1}{n}$ or $\adv \ge \mathrm{Succ}_Q(\tau) + \tfrac{3}{n}$. Observe that $\adv(w,\tau) \ge \mathrm{Succ}_Q(\tau) + \tfrac{3}{n}$ implies that $\adv(w,\tau) \ge \E_B[\adv(w,\tau)] + 1/n$ since $n < N.$

.


\begin{align*}
    & \Pr_{B \leftarrow \mathcal{P}_n} \left[ \adv(w,\tau) \le \mathrm{Succ}_Q(\tau) - \tfrac{1}{n} \text{ OR } \adv(w,\tau) \ge \mathrm{Succ}_Q(\tau) + \tfrac{2}{n} \right] \\
    \leq & 
    \Pr_{B \leftarrow \mathcal{P}_n}\!\left[
    \left|\mathrm{adv}(w,\tau) 
    - \mathbb{E}[\mathrm{adv}(w,\tau)]\right| 
    \geq \frac{1}{n}
    \right] \notag\\
    \leq  &\frac{\mathrm{Var}[\mathrm{adv}(w,\tau)]}{1/n^2}
    \notag\\
    \le & \frac{n^2}{N^2}\left[K_1(\tau)+\frac{K_2(\tau)-K_1(\tau)}{N} \right]+\frac{12n^2}{N}~\mathrm{Succ}^2_Q(\tau/4)\\
    \le & \frac{n^2}{N^2}(7 + \frac{16.2}{N}) + \frac{48n^2}{N}\\
    \le & \frac{7}{N}\frac{n^2}{N} + \frac{16.2}{N^2} \frac{n^2}{N} + 48\frac{n^2}{N} \\
    \le & 48.5\frac{n^2}{N} \tag{For $N \ge 32$}
    \end{align*}
%
\end{proof}

For the next lemma, we use the well-known bound on the CDF of a standard normal random variable $Z$ where $t > 0$.
$$\Phi(-t) = \Pr[Z < -t] = \Pr[Z > t] \le \dfrac{1}{t\sqrt{2\pi}}e^{-t^2/2}$$

\begin{lemma}\label{lemma:lem31}
    For large $n$ and $B$ chosen randomly from $n$-bit permutations, no randomized algorithm with $o(N/\log N)$ queries can solve MAC Fishing with probability $Succ_R + \Omega(1)$.
\end{lemma}

\begin{proof} The proof will be based on contradiction.
    We will prove the claim for a deterministic algorithm $A$ that make exactly $t$ queries, for some $t \in o(N / \log N)$. The result generalizes to randomized algorithms using Yao's minimax principle.

    Let $S \subseteq \F_2^n $ denote the set of $t$ queries made by $A$. For each $(\alpha, \beta)$ define the following:

    \begin{align*}
        C^s_w(\alpha,\beta) & = \tfrac{1}{\sqrt{t}} \sum_{x \in S} (-1)^{\alpha \cdot B(x) \xor \beta \cdot B(x \xor w)}  \text{ (estimate of $C$ based on seen queries, from $S$)}\\
        C^u_w(\alpha,\beta) & = \tfrac{1}{\sqrt{N-t}} \sum_{x \not\in S} (-1)^{\alpha \cdot B(x) \xor \beta \cdot B(x \xor w)}  \text{ (estimate of $C$ based on unseen inputs)}\\
    \end{align*}
\end{proof}

Observe that $N/t = \omega(1/\log N) = \omega(1/\ln N)$. Then, $C_w(\alpha,\beta)$ can be expressed as 
\begin{align}\label{eqn:1}
C_w(\alpha,\beta) = \sqrt{\tfrac{t}{N}} C^s_w(\alpha,\beta) + \sqrt{1 - \tfrac{t}{N}} C^u_w(\alpha, \beta) = C^s_w(\alpha, \beta)/\omega(\sqrt{\ln N}) + C^u_w(\alpha, \beta)(1-o(1))
\end{align}

{\em The following claims are different from Aaronson and Chen since they depend on masked cross-correlation values.}

\begin{claim}
    For a randomly chosen $B$, $|C^s_w(\alpha,\beta)|$ is low for all pairs of $\alpha, \beta$ with high probability.
\end{claim}

Recall that $C^s_w(\alpha,\beta) \sim \N(0,1)$ for $\alpha\not =\beta$ and $\sim \N(0,2)$ for $\alpha=\beta$.

Consider the case when $\alpha \not= \beta$ which can happen for $O(N^2)$ pairs. Then, 
$$\Pr\left[|C^s_w(\alpha,\beta)| \ge 2\sqrt{\ln N}\right] = \tfrac{2}{\sqrt{2\pi}} \int_{-\infty}^{-2\sqrt{\ln N}} e^{-x^2/2} dx = 2\Phi(-2\sqrt{\ln N}) \leq \frac{2e^{-2\ln N}}
    {2\sqrt{\ln N}\sqrt{2\pi}} 
    = \frac{1}{\sqrt{2\pi\ln N}}
    \cdot\frac{1}{N^2} $$

Next, consider the case of $\alpha = \beta$. There are exactly $N$ such pairs. Then too, the above probability can be written as
$$ \Pr\left[|C^s_w(\alpha,\alpha)| \ge 2\sqrt{\ln N}\right] = \tfrac{2}{\sqrt{2\pi}} \int_{-\infty}^{-\sqrt{2\ln N}} e^{-x^2/2} dx = 2\Phi(-\sqrt{2 \ln N}) \leq \frac{2e^{-\ln N}}
    {\sqrt{2\ln N}\sqrt{2\pi}} 
    = \frac{1}{\sqrt{\pi\ln N}}
    \cdot\frac{1}{N}$$

Thus, the probability that there exists some $(\alpha,\beta)$ for which $|C^s_w(\alpha,\beta)| \ge 2\sqrt{\ln N}$ can be upper bounded using union bound.
$$(N^2-N) \cdot \frac{1}{\sqrt{2\pi\ln N}}
    \cdot\frac{1}{N^2} + N \cdot \frac{1}{\sqrt{\pi\ln N}}
    \cdot\frac{1}{N} \leq (1-\tfrac{1}{N})\frac{1}{\sqrt{2\pi\ln N}}
    + \frac{1}{\sqrt{\pi\ln N}}
    = \frac{2}{\sqrt{\pi\ln N}}.$$

This proves the claim that
$$\Pr_B \left[\forall (\alpha, \beta),~ |C^s_w(\alpha,\beta)| \le 2\sqrt{\ln N}\right] \ge 1 - \frac{2}{\sqrt{\pi\ln N}}. $$ 
\qed

Let's define the event $\bad$: for all $(\alpha,\beta)$, $|C^s_w(\alpha,\beta)| \le 2\sqrt{\ln N}$. Therefore, $\Pr_B[not(\bad)] \le \tfrac{2}{\sqrt{\pi\ln N}}$. Now, lets focus on the possible accuracy of unseen part of $C_w(\alpha,\beta)$.

\begin{claim}
    For every $\alpha,\beta$ and a randomly chosen $B$, $C^u_w(\alpha,\beta)$ is low with high probability.
\end{claim}
This is expected since $C^u_w(\alpha,\beta)$ does not depend on the answers of any query made by $A$ irrespective of the pair of masks. Thus, $C^u_w(\alpha,\beta) \sim ~\N(0,1)$ for $\alpha \not= \beta$ and $\sim \N(0,2)$ for $\alpha = \beta$. For both the cases,
\begin{align*}
    \Pr_B \left[ |C^u_w(\alpha,\beta)| \ge 1 - o(1) \right] &= \frac{2}{\sqrt{2\pi}} \int_{1 - o(1)}^{\infty} e^{-x^2/2} dx\\
    & = \frac{2}{\sqrt{2\pi}} \int_{1}^{\infty} e^{-x^2/2} dx + o(1)\\
    & = Succ_R + o(1)
\end{align*}
This proves the claim. \qed

The rest of the proof follows the exact same steps as Aaronson and Chen.

Suppose that $A$ has output some specific $(\alpha,\beta)$. We will show that $A$ most likely miscalculated the masked pairs; the masked cross-correlation for that pair is not high at all.

Let's consider the case that $\bad$ happens for some $B$. We know from Equation~\ref{eqn:1} that for a random $B$ in that case,
$$C_w(\alpha, \beta) \le 2\sqrt{\ln N}/\omega(\sqrt{\ln N}) + C^u_w(\alpha,\beta)(1 - o(1)) = o(1) + C^u_w(\alpha,\beta)(1 - o(1)).$$

Therefore, $\Pr_B \left[|C_w(\alpha,\beta)| \ge 1\right] \le \Pr_B \left[|C^u_w(\alpha,\beta)| \ge 1\right]$.

The right hand side is upper bounded by $Succ_R + o(1)$ and this puts a bound on $(\alpha,\beta)$ being a good mask pair, synonymously, $A$ being successful. Now we combine both the scenarios of $\bad$ happening and not happening.

\begin{align*}
    & \Pr_B[(\alpha,\beta) \text{ is good}]\\
=   & \Pr_B[(\alpha,\beta) \text{ is good} ~|~ \bad] \cdot \Pr_B[\bad] \\
    & + \Pr_B[(\alpha,\beta) \text{ is good} ~|~ not(\bad)] \cdot \Pr_B[not(\bad)] \\
\le & (\mathrm{Succ}_R + o(1)) \cdot 1 + 1 \cdot \tfrac{2}{\sqrt{\pi\ln N}} = \mathrm{Succ}_R + o(1).
\end{align*}


This ends the proof of the lemma.\qed

%% file: final/appendix.tex
\section{Experimental Validation}\label{sec:exp-validation}

We validate the constructions of Section~\ref{sec:applications} on a deliberately
small cipher. For the capacity-family stages every quantity --- the full masked
auto-correlation spectrum, the capacity, and the key-recovery scores --- is computed
exactly by enumeration; for the quantum MAC-fishing stages we additionally simulate
the two primitives the attack relies on. The aim is to confirm that the pipeline
behaves as the theory predicts: that the masked auto-correlations follow the
$\N(0,1)$/$\N(0,2)$ baseline of Section~\ref{sec:background}, that the capacity
separates the cipher from a random permutation at the predicted data complexity, and
that the signed template recovers the last-round subkey where the unsigned energy
cannot. We do not run quantum hardware: \textsc{MACSample} is \emph{simulated} by
computing the exact distribution $\CD_w$ (the normalized spectrum
$\cor(B,w;\alpha,\beta)^2$) and sampling $\propto\cor^2$ from it, and
\textsc{QuantumVerify} by perturbing the exact signed correlation with additive
Gaussian noise of width $\varepsilon$ (its amplitude-estimation precision), so what
we validate are the statistics Algorithms~\ref{algo:cd_w}
and~\ref{alg:quantum-verify} would produce. The signed reference $v^\star$ is read
from the reduced cipher under the known key, the known-key validation setting of
Section~\ref{sec:cap-kr}.

\subsection{A toy Simon-like cipher}\label{sec:exp-cipher}

We use a Feistel cipher in the style of Simon with $6$-bit words, hence block size
$n=12$ ($N=2^{12}=4096$). Writing the state as $(L,R)\in\F_2^6\times\F_2^6$ and
$S^j$ for left rotation by $j$ within a word, one round maps
\[
  (L,R)\ \longmapsto\ \bigl(F(L)\oplus R\oplus k,\ L\bigr),
  \qquad
  F(x) = (S^1 x \wedge S^3 x)\oplus S^2 x ,
\]
with a $6$-bit round key $k$ (the rotation constants $(1,3,2)$ are the Simon
pattern $(1,n_w/2,2)$ scaled to a $6$-bit word $n_w=6$). The reduced cipher $B$ is
$r=4$ rounds under fixed round keys; the target $E_K$ appends one further round
under a secret last-round subkey $k^\star$, so $E_K=R_{k^\star}\circ B$ and the
last-round subkey space is $\kappa=6$ bits ($64$ candidates). All correlations
below are the normalized $\cor(B,w;\alpha,\beta)=2^{-n}\sum_x(-1)^{\alpha\cdot
B(x)\oplus\beta\cdot B(x\oplus w)}\in[-1,1]$ of Section~\ref{sec:mdl}; for a random
permutation $\cor\sim\N(0,2^{-n})$, equivalently $C^w=\sqrt N\,\cor\sim\N(0,1)$.

\subsection{Masked auto-correlation spectrum and sampling}\label{sec:exp-spectrum}

For a fixed difference $w$ the entire spectrum $\{C^w(\alpha,\beta)\}_{\alpha,\beta}$
is the Walsh--Hadamard transform of the distribution of
$x\mapsto(B(x),B(x\oplus w))$, so all $2^{2n}$ values are obtained exactly in one
transform. Sweeping $w$ and ranking by the strongest non-trivial pair selects
$w=\mathtt{0x8f}$. Figure~\ref{fig:spectrum} plots the spectrum of a \emph{random}
permutation at this $w$: it matches the predicted baseline of
Section~\ref{sec:background} --- $\N(0,1)$ for off-diagonal pairs $\alpha\neq\beta$
and the heavier $\N(0,2)$ for diagonal pairs $\alpha=\beta$. Against this baseline
the reduced cipher concentrates its correlation in a handful of heavy pairs, listed
in Table~\ref{tab:heavy} and reaching $|C^w|\approx 24$, far into the tails; the
heaviest are diagonal pairs, exactly the differential-linear (DLCT) entries of
Section~\ref{sec:mdl}. Figure~\ref{fig:macsample} illustrates the simulated quantum
\textsc{MACSample}. Sampling $\propto\cor^2$ over the \emph{full} spectrum is
bulk-dominated --- by Parseval the $2^{2n}$ pairs each carry $\cor^2\sim 1/N$, so
the heavy pairs are a vanishing fraction of the total mass --- which is precisely
why the quantum sampler relies on amplitude amplification to concentrate on the
heavy set $V_w=\{\cor^2\ge\tau\}$. Conditioned on $V_w$ the returned pair is
distributed $\propto\cor^2$, which the empirical draws reproduce, as
Algorithm~\ref{algo:cd_w} prescribes.

\begin{minipage}[t]{0.48\textwidth}
    \centering
    \includegraphics[width=\linewidth]{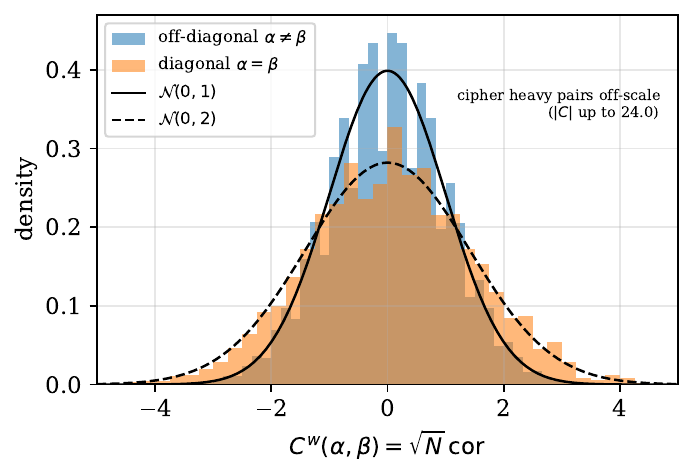}
    \captionof{figure}{Masked auto-correlation spectrum at
      $w=\mathtt{0x8f}$}
    \label{fig:spectrum}
\end{minipage}
\hfill
\begin{minipage}[t]{0.48\textwidth}
    \centering
    \includegraphics[width=\linewidth]{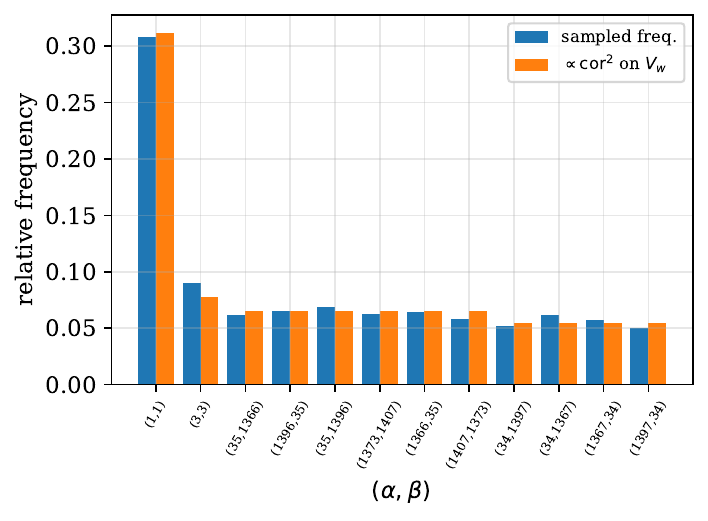}
    \captionof{figure}{Simulated quantum \textsc{MACSample} conditioned on the heavy set
      $V_w$.}
    \label{fig:macsample}
\end{minipage}


\begin{table}[h]
\centering
\begin{tabular}{lllrr}
\hline
$w$ & $\alpha$ & $\beta$ & $\mathrm{cor}$ & $N\,\mathrm{cor}^2$ \\
\hline
0x8f & 0x01 & 0x01 & -0.3750 & 576.0 \\
0x8f & 0x03 & 0x03 & +0.1875 & 144.0 \\
0x8f & 0x23 & 0x56 & +0.1719 & 121.0 \\
0x8f & 0x74 & 0x23 & +0.1719 & 121.0 \\
0x8f & 0x23 & 0x74 & +0.1719 & 121.0 \\
0x8f & 0x5d & 0x7f & +0.1719 & 121.0 \\
\hline
\end{tabular}
\caption{Heaviest masked auto-correlation tuples of the reduced toy cipher (baseline $N\,\mathrm{cor}^2\approx 1$).}
\label{tab:heavy}
\end{table}


\subsection{The capacity distinguisher}\label{sec:exp-capacity}

For the capacity distinguisher we stack $t=4$ pairs into the projected function
$g^w$. The heaviest pairs are the diagonal DLCT entries
of Table~\ref{tab:heavy}, but a diagonal pair carries the variance-$2$ baseline of
Section~\ref{sec:background}, and diagonal linear combinations would inflate the
degrees of freedom of the null; we therefore take the heaviest \emph{off-diagonal}
pairs whose masks are linearly independent, so that every nonzero combination is
off-diagonal and the null is a clean $\chi^2_{2^t-1}$. This gives capacity
$\Cap(g^w)=0.128$ for the cipher ($N\Cap\approx 523$) against the random baseline
$(2^t-1)/N\approx 0.0037$. Figure~\ref{fig:capacity} plots the statistic
$N\cdot\widehat{\Cap}(g^w)$ over $4000$ random permutations: it matches the
predicted $\chi^2_{2^t-1}$ density (empirical mean $15.1$ against $2^t-1=15$), while
the cipher lies far beyond the entire null. Figure~\ref{fig:datacomplexity} shows the empirical distinguishing success against
the number of plaintext pairs $N$; reliable distinguishing sets in at a small
constant multiple of the predicted scaling $\sqrt{2^t}/\Cap\approx 31$ pairs (dashed
line), consistent with the $N=\Omega(\sqrt{2^t}/\Cap)$ law of
Equation~\eqref{eq:data-dist}.

\begin{minipage}[t]{0.42\textwidth}
    \centering
\centering
\includegraphics[width=0.62\linewidth]{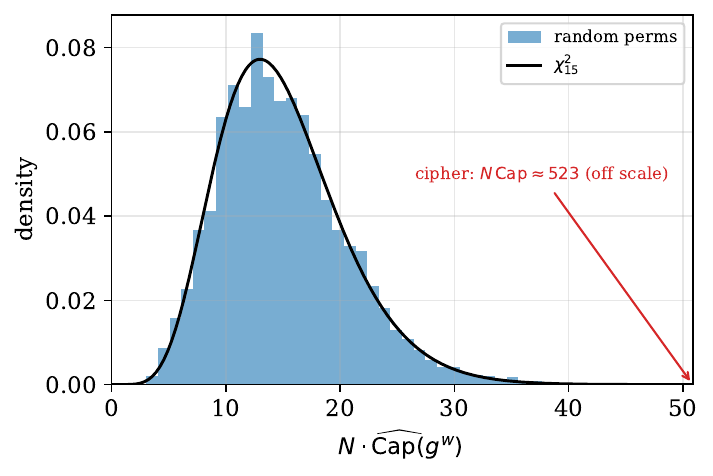}
\captionof{figure}{Capacity distinguisher at $w=\mathtt{0x8f}$, $t=4$ off-diagonal
  independent pairs. Histogram of $N\cdot\widehat{\Cap}(g^w)$ over random
  permutations with the $\chi^2_{2^t-1}$ density; the cipher ($N\Cap\approx523$)
  is far off-scale to the right.}
\label{fig:capacity}
\end{minipage}
\hfill
\begin{minipage}[t]{0.42\textwidth}
    \centering
\centering
\includegraphics[width=0.62\linewidth]{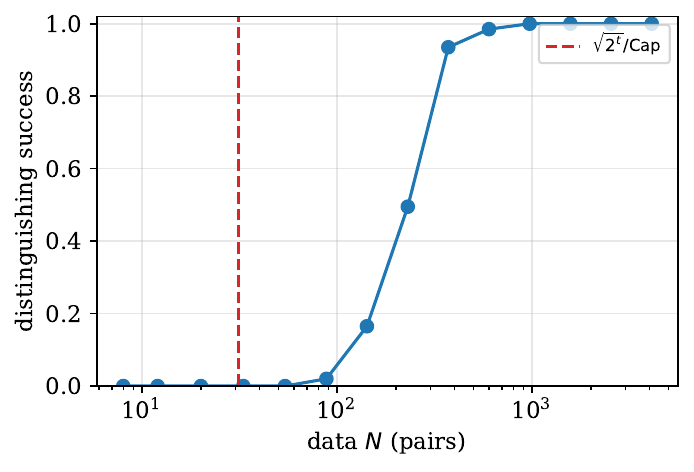}
\captionof{figure}{Distinguishing success versus data $N$. The dashed line marks the
  predicted complexity $\sqrt{2^t}/\Cap(g^w)$.}
\label{fig:datacomplexity}
\end{minipage}

\subsection{Last-round key recovery}\label{sec:exp-keyrecovery}

Finally we recover the last-round subkey of the $5$-round cipher
$E_K=R_{k^\star}\circ B$. Because the Simon subkey enters the masked output bits as
a constant phase, $\cor_i(k)=(-1)^{\langle k,m_i\rangle}\cor_i^\star$ with key-mask
$m_i=(\alpha_i\oplus\beta_i)$ on the subkey word, so the magnitudes
$|\cor_i(k)|$ --- and hence the energy $\widehat{\adv}(k)=\tfrac1m\sum_i\cor_i(k)^2$
--- are \emph{exactly} identical for all $64$ candidates (measured spread $0$).
The energy screen therefore carries no key information whatsoever: all of it lives
in the signs. The diagonal pairs that dominated the distinguisher have
$m_i=0$ and are key-blind; we build the key-recovery template instead from $m=6$
off-diagonal heavy pairs whose key-masks are linearly independent
($\mathrm{rank}\,\{m_i\}=6=\kappa$). With this template the signed cosine score
$\mathrm{Score}(k)$ of Equation~\eqref{eq:cosine} is maximised on the single coset
$k^\star\oplus\langle m_i\rangle^\perp=\{k^\star\}$, recovering the true subkey
$k^\star=2$ uniquely. Figure~\ref{fig:keyrecovery} contrasts the flat energy
profile with the sharply peaked signed score, and Table~\ref{tab:keyrec} ranks the
top candidates. Had the template masks failed to span, the score would instead
recover $k^\star$ up to the coset $\langle m_i\rangle^\perp$ --- the residual
ambiguity is exactly $\kappa-\mathrm{rank}\,\{m_i\}$ bits, the standard
key-information limit of a masked linear attack.

\begin{figure}[h]
\centering
\includegraphics[width=0.85\linewidth]{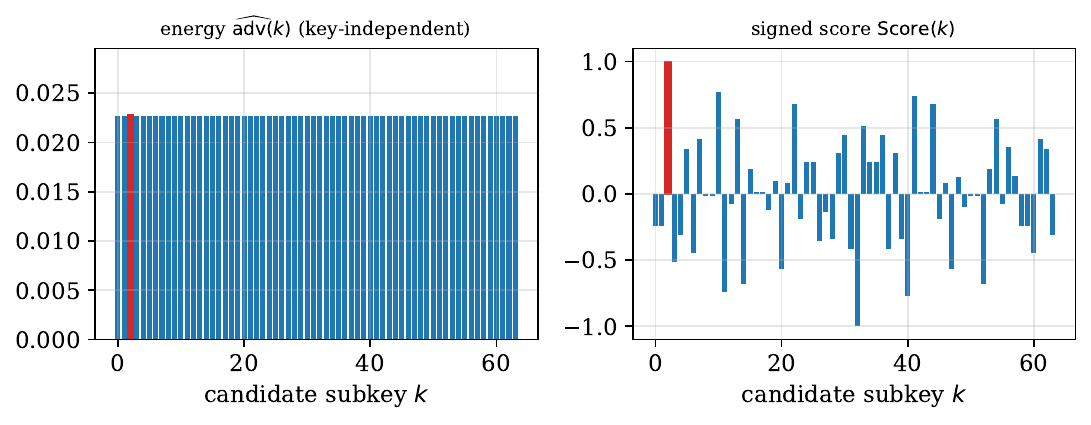}
\caption{Last-round key recovery over the $64$ candidate subkeys. Left: the energy
  $\widehat{\adv}(k)$ is identical for every candidate. Right: the signed score
  peaks uniquely at the true subkey $k^\star=2$ (highlighted).}
\label{fig:keyrecovery}
\end{figure}

\begin{table}[h]
\centering
\begin{tabular}{rrrl}
\hline
$k$ & $\widehat{\mathrm{adv}}(k)$ & $\mathrm{Score}(k)$ &  \\
\hline
2 & 0.0227 & +1.0000 & $\leftarrow k^\star$ \\
10 & 0.0227 & +0.7709 &  \\
41 & 0.0227 & +0.7414 &  \\
22 & 0.0227 & +0.6770 &  \\
44 & 0.0227 & +0.6770 &  \\
54 & 0.0227 & +0.5669 &  \\
\hline
\end{tabular}
\caption{Last-round key recovery. The energy $\widehat{\mathrm{adv}}(k)$ is identical for all candidates (key-independent magnitudes), so the screen alone is uninformative; the signed score singles out $k^\star$ when the template key-masks span the subkey space.}
\label{tab:keyrec}
\end{table}

\subsection{The MAC-fishing distinguisher}\label{sec:exp-macfish-dist}

The capacity stages above used the classical pipeline of
Section~\ref{sec:cap-attacks}: enumerate the spectrum, rank, aggregate. We now
validate the quantum MAC-fishing pipeline of Section~\ref{sec:mac-attacks}, which
replaces enumeration by sampling and exact correlations by amplitude-estimated
ones, simulating its two primitives classically. \textsc{MACSample} returns a pair
drawn $\propto\cor^2$ from the heavy set $V_w^{N\tau}=\{\cor^2\ge\tau\}$ (the
amplitude-amplified output of Section~\ref{sec:exp-spectrum}); at threshold
$\tau=10^{-2}$ the heavy set has $|V_w^{N\tau}|=145$ pairs and density
$\adv(w,N\tau)=|V_w^{N\tau}|/2^{2n}\approx 8.6\times10^{-6}$, whose inverse is the
amplification overhead. \textsc{QuantumVerify} (Algorithm~\ref{alg:quantum-verify})
returns a \emph{signed} estimate of $\cor(\alpha,\beta)$ to additive precision
$\varepsilon$ in $O(1/\varepsilon)$ queries; resolving membership in $V_w^{N\tau}$,
i.e.\ $|\cor|\ge\sqrt\tau$, fixes the budget at $\varepsilon\sim\sqrt\tau$
($\approx10$ queries here). We model it as $\widehat\cor=\cor+\N(0,\varepsilon^2)$
and study the pipeline as $\varepsilon$ varies.

The \textsc{MACFishingDistinguisher} (Algorithm~\ref{alg:mac-fishing-dist}) draws
$m=16$ pairs, verifies each, and thresholds the averaged energy
$\tfrac1m\sum_i\widehat\cor_i^2$. For the cipher these are genuine heavy pairs; a
random permutation has no heavy set, so its verified correlations are
$\cor\sim\N(0,1/N)$ plus the $\varepsilon$ noise floor.
Figure~\ref{fig:macfish-dist} (left) shows the two statistic distributions at
$\varepsilon=\sqrt\tau$: they separate cleanly, and at a $1\%$ false-positive
operating point the cipher is detected with probability $0.88$, in line with the
Gaussian union bound of Proposition~\ref{prop:mac-dist}. The right panel sweeps the
verification precision: detection is near-certain for $\varepsilon\lesssim\sqrt\tau$
and degrades once the noise floor $\varepsilon^2$ overtakes the heavy energy.

\begin{figure}[h]
\centering
\includegraphics[width=\linewidth]{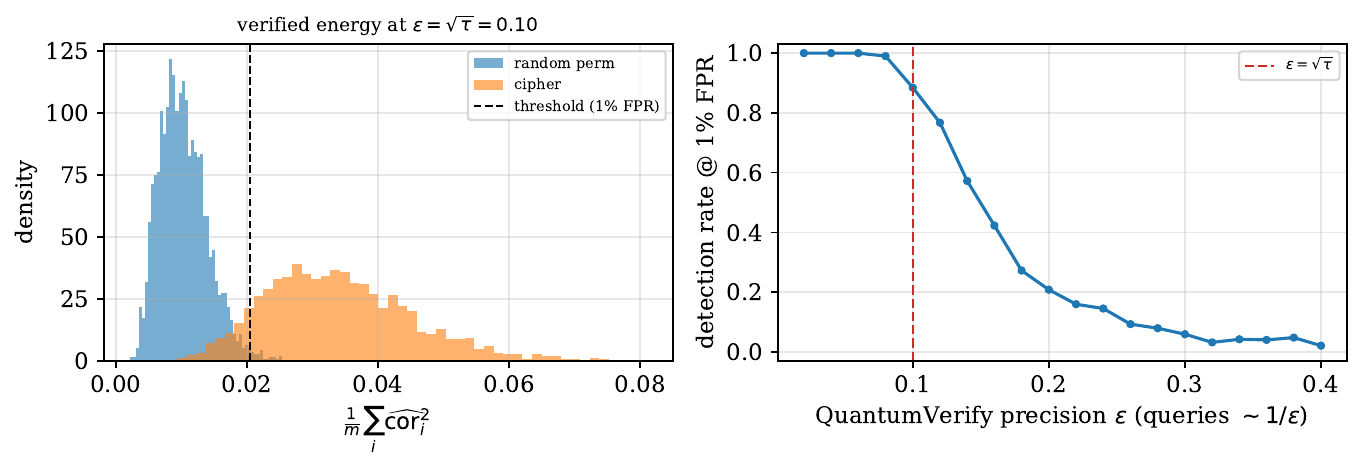}
\caption{MAC-fishing distinguisher. Left: averaged verified energy
  $\tfrac1m\sum_i\widehat\cor_i^2$ for the cipher and for random permutations at
  $\varepsilon=\sqrt\tau$, with the $1\%$ false-positive threshold. Right: detection
  rate versus \textsc{QuantumVerify} precision $\varepsilon$ (query cost
  $\sim1/\varepsilon$); the minimal budget $\varepsilon=\sqrt\tau$ is marked.}
\label{fig:macfish-dist}
\end{figure}

\subsection{MAC-fishing key recovery}\label{sec:exp-macfish-kr}

Key recovery reuses the spanning template of Section~\ref{sec:exp-keyrecovery}
($m=6$ off-diagonal pairs, $\mathrm{rank}\,\{m_i\}=6=\kappa$) but now obtains its
correlations from \textsc{QuantumVerify}. \textsc{TemplateEstablishment}
(Algorithm~\ref{alg:template-establishment}) fixes the same-key signed reference
$v^\star$ offline at high precision; for each subkey candidate,
\textsc{MACFishingKeyRecovery} (Algorithm~\ref{alg:mac-fishing-kr}) verifies the
template on the peeled cipher and scores it against $v^\star$ by the signed cosine
of Equation~\eqref{eq:cosine-mac}. As before the energy is exactly key-blind
(measured spread $0$); recovery rests entirely on the verified signs.

Figure~\ref{fig:macfish-keyrec} (left) shows the score profile: at an adequate
budget the true subkey $k^\star=2$ is the unique maximum. The right panel sweeps
$\varepsilon$ and exposes a precision requirement sharper than the distinguisher's:
because recovery must resolve correlation \emph{signs} rather than aggregate energy,
it needs $\varepsilon$ below the template magnitude $|\cor|\approx0.15$. At the
distinguisher's budget $\varepsilon=\sqrt\tau$ recovery is only marginal ($0.68$),
whereas halving $\varepsilon$ (doubling the query cost) makes it reliable
($\approx1$). The gap is structural: the heaviest correlations of this cipher span
only a rank-$2$ subspace of the subkey-mask space, so four of the six key directions
ride on weaker correlations that demand finer amplitude estimation.
The Grover-style search of
Remark~\ref{grrem} marks on this same signed score and, as there, returns the
recovered coset rather than a unique $k^\star$ unless the template masks span.

\begin{figure}[h]
\centering
\includegraphics[width=\linewidth]{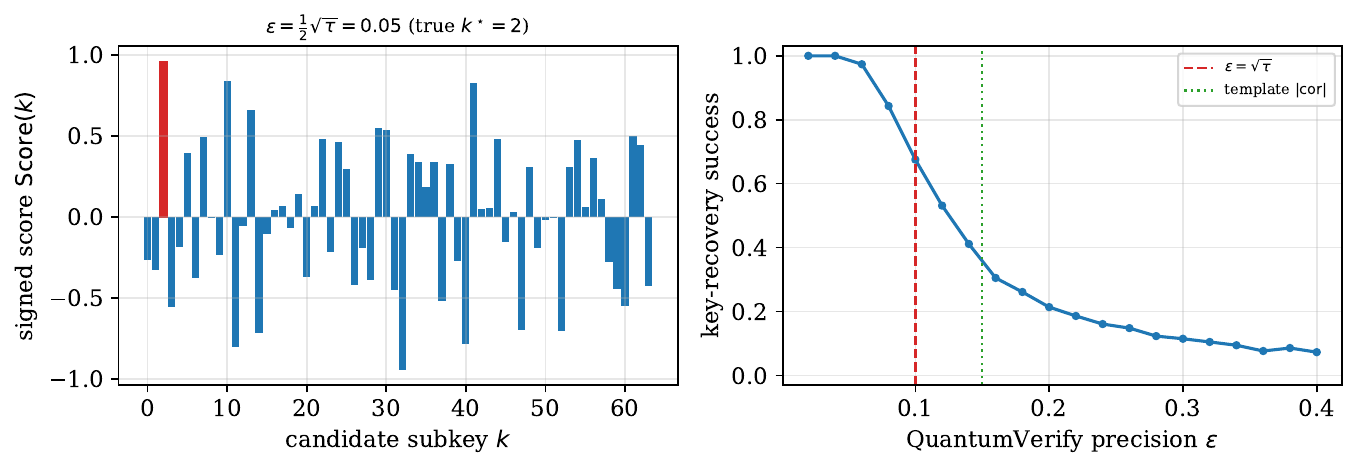}
\caption{MAC-fishing key recovery over the $64$ subkeys. Left: signed score
  $\mathrm{Score}(k)$ at $\varepsilon=\tfrac12\sqrt\tau$, peaking uniquely at
  $k^\star=2$. Right: recovery success versus \textsc{QuantumVerify} precision
  $\varepsilon$, with $\sqrt\tau$ and the template magnitude $|\cor|$ marked.}
\label{fig:macfish-keyrec}
\end{figure}


These experiments reproduce, on a genuine cipher structure, the two phenomena that
shape Section~\ref{sec:applications}: the capacity aggregates the heavy masked
auto-correlations into a clean distinguisher at the predicted data complexity, and
the last-round subkey is invisible to the unsigned energy but recovered exactly by
the signed template once its masks span the subkey space. The MAC-fishing stages
reproduce the same two attacks under the quantum primitives of
Section~\ref{sec:mac-attacks} --- sampled heavy pairs and amplitude-estimated signed
correlations --- and expose a precision hierarchy between them: distinguishing
tolerates the minimal $\varepsilon\sim\sqrt\tau$ budget, while sign-based key
recovery requires finer verification.

\section{Classical Verification and Scoring for MAC Fishing}\label{app:classical-mac}

These are the classical Phase-2 procedures for the MAC-Fishing attacks of
Section~\ref{sec:mac-kr}. 

\begin{algorithm}[h]
\caption{\textsc{ClassicalVerify}}
\label{alg:classical-verify}
\begin{algorithmic}[1]
\Require oracle $B:\mathbb{F}_2^n\to\mathbb{F}_2^n$; difference $w$; candidate
  mask pair $(\alpha,\beta)$; sample count $N$
\Ensure estimated squared correlation $\widehat{C}_w(\alpha,\beta)^2$
\State $\mathit{acc}\gets 0$
\For{$j = 1$ \textbf{to} $N$}
  \State Choose $x_j\in\mathbb{F}_2^n$ uniformly at random
  \State $y_j\gets B(x_j)$, \quad $y'_j\gets B(x_j\oplus w)$
  \State $\mathit{acc}\gets \mathit{acc} + (-1)^{\alpha\cdot y_j\oplus\beta\cdot y'_j}$
\EndFor
\State $\widehat{C}_w(\alpha,\beta)\gets \mathit{acc}/N$
\State \Return $\widehat{C}_w(\alpha,\beta)^2$
\end{algorithmic}
\end{algorithm}


\begin{algorithm}[h]
\caption{\textsc{ClassicalKeyScore}}
\label{alg:classical-key-score}
\begin{algorithmic}[1]
\Require plaintext--ciphertext oracle $E_K$; difference $w$; key guess $k$;
  template $\{(\alpha_i,\beta_i)\}_{i=1}^m$; reference vector $v^\star$;
  sample count $N$
\Ensure scores $\widehat{\mathrm{adv}}(k)$, $T(k)$, $\mathrm{Score}(k)$
\State $v(k)\gets ()$
\For{$i = 1$ \textbf{to} $m$}
  \State $\mathit{acc}\gets 0$
  \For{$j = 1$ \textbf{to} $N$}
    \State Choose $x_j\in\mathbb{F}_2^n$ uniformly at random
    \State $c_j\gets E_K(x_j)$, \quad $c'_j\gets E_K(x_j\oplus w)$
    \State $y_j\gets R_k^{-1}(c_j)$, \quad $y'_j\gets R_k^{-1}(c'_j)$
      \Comment{partial decryption under $k$}
    \State $\mathit{acc}\gets \mathit{acc} +
      (-1)^{\alpha_i\cdot y_j\oplus\beta_i\cdot y'_j}$
  \EndFor
  \State append $\mathit{acc}/N$ to $v(k)$
\EndFor
\State $\widehat{\mathrm{adv}}(k)\gets \tfrac{1}{m}\sum_i v_i(k)^2$
\State $T(k)\gets \langle v(k),v^\star\rangle$
\State $\mathrm{Score}(k)\gets T(k)/(\|v(k)\|\,\|v^\star\|)$
\State \Return $\widehat{\mathrm{adv}}(k)$, $T(k)$, $\mathrm{Score}(k)$
\end{algorithmic}
\end{algorithm}

\section{Use of Wilson's Score}\label{app:Wilson}

\paragraph{The estimation problem.}
Each candidate approximation has an unknown \emph{support probability}
\[
  \rho \;=\; \Pr_{k}\!\bigl[\,\lvert C_k\rvert > Z\sigma\,\bigr],
\]
the fraction of keys for which its correlation $C_k$ clears the
random-permutation floor (here $\sigma=2^{-n/2}$ and $Z=2$). From a sample of
$K$ keys we observe $X$ successes and the point estimate $\phat=X/K$; we want
a \emph{one-sided lower confidence bound} on $\rho$ so that we can keep a pair
only when we are confident its true support exceeds a threshold.
 
\paragraph{The Wilson score interval.}
The Wilson interval is obtained by inverting the score test rather than the
Wald test. Testing $H_0:\rho=p$ with the score statistic
$(\phat-p)\big/\sqrt{p(1-p)/K}$ and collecting all $p$ not rejected at level
$z$ gives
\[
  \bigl\{\,p:\ (\phat-p)^2 \le z^2\,p(1-p)/K\,\bigr\},
\]
a quadratic inequality in $p$,
\[
  \Bigl(1+\tfrac{z^2}{K}\Bigr)p^2
  -\Bigl(2\phat+\tfrac{z^2}{K}\Bigr)p
  +\phat^2 \;\le\; 0 .
\]
Solving the quadratic yields an interval centred at
$\bigl(\phat+\tfrac{z^2}{2K}\bigr)\big/\bigl(1+\tfrac{z^2}{K}\bigr)$ with the
half-width below; the \emph{one-sided lower bound} is
 
\begin{equation}
  \boxed{\;
  \rho_{\mathrm{LB}}
  \;=\;
  \frac{\phat+\dfrac{z^2}{2K}}{1+\dfrac{z^2}{K}}
  \;-\;
  \frac{z}{1+\dfrac{z^2}{K}}
  \sqrt{\dfrac{\phat(1-\phat)}{K}+\dfrac{z^2}{4K^2}}
  \;}
  \label{eq:wilson}
\end{equation}
(clamped at $0$). Three properties matter for the acceptance criteria: the bound always lies
in $[0,1]$; its small-sample coverage is far better than Wald's; and the
$z^2/K$ terms pull the estimate toward $1/2$, i.e.\ the bound is
\emph{conservative} for a high target $\rho$ which makes it conservative when accepting or rejecting against a threshold. The experiment uses
$z=\Phi^{-1}(0.95)=1.6449$ (a one-sided $95\%$ bound, appropriate for a
one-sided ``$\rho\ge\tau$'' claim) and accepts a pair only when
$\rho_{\mathrm{LB}}\ge 0.70$.

Several of our experiments share one statistical structure. Each estimates the
fraction of a population too large to enumerate either the key space, or a part of $2^{2n}$ total mask pairs. On these differential-linear approximation needs to clears a
fixed threshold, where either the key or the mask-pairs are sampled uniformly at random and counting successes. Every draw
is classified by a cross-correlation evaluated \emph{exactly} over all $2^n$
inputs rather than sub-sampled, so the per-draw outcome is noise-free and the
success count is a clean binomial of i.i.d.\ trials; a Wilson interval therefore
applies directly to the proportion and, after rescaling, to
the population count itself. The estimate stays informative only because each
threshold is pinned to the random-permutation noise floor (per-correlation
variance ${\approx}\,1/N$), fixing the qualifying fraction at a
modest tail mass independent of block size instead of letting it vanish
exponentially, so a uniform sample reliably lands in the set. Since the interval
captures only the sampling uncertainty of a single estimate, any spread observed
across keys or input differences is separated from it and read as genuine
population variation against that same baseline.

%% file: final/appendix-simple-proofs.tex
\section{Proof of Lemma~\ref{moments}}\label{app:lem1}

\lemsimple*

\begin{proof}

The first two cases follow directly from the fact that if $X \sim \N(0,\sigma^2)$, then $X^2/\sigma^2$ follows a $\chi^2(1)$ distribution with degree of freedom 1, and that has mean 1 and variance 2.

For the 3rd case, we will simply average the above values over all $(\alpha,\beta) \in \F_2^n \times \F_2^n$; among those pairs, there are $N$ pairs with $\alpha = \beta$ and $N^2-N$ pairs with $\alpha \neq \beta$. We will also require $\E_B[C_w(\alpha, \beta)^4]$ which can be obtained from $\Var_B[C_w(\alpha, \beta)^2]$ and $\E_B[C_w(\alpha,\beta)^2]$.
\begin{itemize}
    \item For $\alpha \not=\beta$, $\E_B[C_w(\alpha, \beta)^4] = \Var_B[C_w(\alpha, \beta)^2] + \E_B[C_w(\alpha,\beta)^2]^2 = 2 + 1 = 3$.
    \item For $\alpha =\beta$, $\E_B[C_w(\alpha, \beta)^4] = \Var_B[C_w(\alpha, \beta)^2] + \E_B[C_w(\alpha,\beta)^2]^2 = 8 + 2^2 = 12$.
\end{itemize}

\begin{align*}
    \mathbb{E}_{\alpha,\beta}\mathbb{E}_B[C_w(\alpha,\beta)^2] 
    &= \frac{N-1}{N} \cdot 1 + \frac{1}{N} \cdot 2 
     = 1 + \frac{1}{N}.
\end{align*}
\begin{align*}
    \mathbb{E}_{\alpha,\beta}\mathbb{E}_B[C_w(\alpha,\beta)^4] 
    &= \frac{N-1}{N} \cdot 3 + \frac{1}{N} \cdot 12 
     = 3 + \frac{9}{N}.
\end{align*}
\begin{align*}
    \mathrm{Var}_{\alpha,\beta,B}[C_w(\alpha,\beta)^2] 
    &= \mathbb{E}_{\alpha,\beta}\mathbb{E}_B[C_w(\alpha,\beta)^4] 
     - \left(\mathbb{E}_{\alpha,\beta}\mathbb{E}_B[C_w(\alpha,\beta)^2]\right)^2 \notag\\
    &= \left(3 + \frac{9}{N}\right) - \left(1 + \frac{1}{N}\right)^2 \notag\\
    &= 2 + \frac{7}{N} - \frac{1}{N^2}.
\end{align*}

\end{proof}

\section{Proof of Lemma~\ref{lem:vtaulemma}}
\label{sec:vtaulemmasection}

\vtaulemma*

\begin{proof}
    By Definition~\ref{def:heavy_pairs}, we have
    \begin{equation*}
    V_w^\tau = \{(\alpha,\beta)\in\mathbb{F}_2^n\times\mathbb{F}_2^n : C_w(\alpha,\beta)^2\geq\tau\}
\end{equation*}
So, 
\begin{equation*}
    |V_w^\tau| = \sum_{(\alpha,\beta)}\delta(\alpha,\beta), 
    \qquad 
    \text{where } \delta(\alpha,\beta) = \mathbb{I}[C_w(\alpha,\beta)^2\geq\tau].
\end{equation*}

Note that $(\alpha,\beta)\in V_w^\tau \iff C_w(\alpha,\beta)^2 \geq\tau$. So:
\begin{equation*}
    p_{\alpha,\beta} 
    = \Pr_B\left[(\alpha,\beta)\in V_w^\tau\right] 
    = \Pr_B\left[C_w(\alpha,\beta)^2\geq\tau\right].
\end{equation*}
\paragraph{Case 1: $\alpha\neq\beta$.} 
$C_w(\alpha,\beta)\sim\mathcal{N}(0,1)$\\

\begin{align*}
    p_{\alpha,\beta}&= \Pr_B\left[C_w(\alpha,\beta)^2\geq\tau\right]\\
    &= \Pr_B\left[|C_w(\alpha,\beta)|\geq \sqrt{\tau}\right]\\
    &=2\Pr_B\left[C_w(\alpha,\beta)\geq\sqrt{\tau}\right] ~~~\text{since,}~\mathcal{N}(0,1)~\text{is symmetric around}~0.\\
    &=2\left(1-\Pr_B\left[C_w(\alpha,\beta)\leq\sqrt{\tau}\right]\right)\\
    &=2\left(1-\Phi(\sqrt{\tau})  \right) \\
    &=2\Phi(-\sqrt{\tau})\\
\end{align*}

\paragraph{Case 2: $\alpha=\beta\neq 0^n$.}

$C_w(\alpha,\alpha)\sim\mathcal{N}(0,2)$. So, we can write $C_w(\alpha,\alpha)$ in terms of standard normal random variable as, $C_w(\alpha,\alpha)=\sqrt{2}Z,$ where $Z \sim \mathcal{N}(0,1).$\\

\begin{align*}
    p_{\alpha,\alpha} 
    &= \Pr_B\left[C_w(\alpha,\alpha)^2\geq\tau\right] \notag\\
    &= \Pr_B\left[2Z^2\geq\tau\right] \notag\\
    &= \Pr_B\left[Z^2\geq\frac{\tau}{2}\right] \notag\\
    &= \Pr_B\left[|Z|\geq\sqrt{\frac{\tau}{2}}\right] \notag\\
    &= 2\Pr_B\left[Z\geq\sqrt{\frac{\tau}{2}}\right] \notag\\
    &= 2\Phi\!\left(-\sqrt{\frac{\tau}{2}}\right)
\end{align*}
Combining both cases, we get
\begin{equation*}
    p_{\alpha,\beta} = 
    \begin{cases}
        2\Phi(-\sqrt{\tau}) & \alpha\neq\beta\\
        2\Phi\!\left(-\sqrt{\tau/2}\right) & \alpha=\beta\neq 0^n
    \end{cases}
\end{equation*}

\paragraph{Expectation of $|V_w^\tau|$}

\begin{align*}
    \mathbb{E}_B\left[|V_w^\tau|\right] 
    &= N\cdot 2\Phi\!\left(-\sqrt{\tau/2}\right) 
    + (N^2-N)\cdot 2\Phi(-\sqrt{\tau}) \notag\\
    &= 2N^2\Phi(-\sqrt{\tau}) 
    + 2N\Phi\!\left(-\sqrt{\tau/2}\right) 
    - 2N\Phi(-\sqrt{\tau}) \notag\\
    &= 2N^2\Phi(-\sqrt{\tau}) 
    + 2N\left(\Phi\!\left(-\sqrt{\tau/2}\right) 
    - \Phi(-\sqrt{\tau})\right)\\
    &=N^2 \cdot p_1 + N \cdot (p_2-p_1)
\end{align*}

where, $p_1 = 2\Phi(-\sqrt{\tau})$ and $p_2 = 2\Phi(-\sqrt{\tau/2})$.

\paragraph{Variance of $|V_w^\tau|$}

$C_w(\alpha,\beta)$ are independent for different pairs of $(\alpha, \beta)$, so the indicator functions $\delta(\alpha, \beta)$ are independent across different mask pairs $(\alpha, \beta)$ and thus the covariance terms vanish. Using the fact that the variance of a Bernoulli$(p)$ random variable is equal to $p(1-p),$ we get the variance of $|V_w^\tau|$ as

\begin{equation*}
    \mathrm{Var}_B\left[|V_w^\tau|\right] 
    = (N^2-N)\cdot p_1(1-p_1) + N\cdot p_2(1-p_2).
\end{equation*}

\end{proof}